\documentclass{article}

\usepackage{multirow}
\usepackage{listings}
\usepackage{hyperref}
\usepackage{amsfonts}
\usepackage{amsmath}
\usepackage{amsthm}
\usepackage[capitalise, noabbrev]{cleveref}
\usepackage{mathpartir}
\usepackage{stmaryrd} 
\usepackage{tikz}
\usetikzlibrary{calc,positioning,fit,backgrounds}
\definecolor{cbBlue}{RGB}{0,114,178}
\pgfdeclarelayer{bgOuter}
\pgfdeclarelayer{bgForml}
\pgfsetlayers{bgOuter,bgForml,main}

\newcommand\graphalg{\textsf{GraphAlg}}
\DeclareMathOperator{\pickAny}{\mathsf{pickAny}}
\DeclareMathOperator{\cast}{\mathsf{cast}}

\newcommand{\Idim}{\mathcal{D}}
\newcommand{\Imat}{\mathsf{mat}}

\newcommand{\ringBool}{\mathbb{B}_{\lor,\land}}
\newcommand{\ringInt}{\mathbb{Z}_{+,\times}}
\newcommand{\ringReal}{\mathbb{R}_{+,\times}}
\newcommand{\ringTropInt}{\mathbb{Z}_{\min,+}}
\newcommand{\ringTropReal}{\mathbb{R}_{\min,+}}
\newcommand{\ringTropMaxInt}{\mathbb{Z}_{\max,+}}
\newcommand{\ringTropMaxReal}{\mathbb{R}_{\max,+}}

\newcommand{\encoded}[1]{\langle #1 \rangle}

\newtheorem{theorem}{Theorem}
\newtheorem{example}{Example}
\newtheorem{lemma}{Lemma}
\newtheorem{corollary}{Corollary}

\newcommand{\ML}{\mathsf{MATLANG}}
\newcommand{\transp}[1]{#1^T}
\newcommand{\SyntaxStyle}{\mathsf}
\newcommand{\LetIn}[3]{\SyntaxStyle{let}\ #1=#2\ \SyntaxStyle{in}\ #3}

\DeclareMathOperator{\diag}{\mathsf{diag}}
\DeclareMathOperator{\Apply}{\mathsf{apply}}

\newcommand{\one}{\mathbf{1}}

\newcommand{\M}{M}

\newcommand{\ffor}[3]{\texttt{for}\, #1,#2 \texttt{.}\, #3}
\newcommand{\Mvar}{\mathcal{V}}
\newcommand{\I}{\mathcal{I}}
\newcommand{\Sch}{\mathcal{S}}
\newcommand{\ttype}{\textsf{type}}
\newcommand{\ttypes}{\textsf{type}_{\Sch}}
\newcommand{\sem}[2]{\llbracket #1 \rrbracket(#2)}
\newcommand{\Mnam}{\mathcal{M}}

\lstdefinelanguage{GraphAlg}{
  morekeywords={func, for, in, until, return}, 
  morekeywords={Matrix, Vector, bool, int, real, trop\_int, trop\_real, trop\_max\_int}, 
  morekeywords={T, diag, apply, select, tril, triu, reduceRows, reduceCols, reduce, pickAny, zero, one, cast}, 
  morekeywords={nrows, cols, nvals}, 
  keywordstyle=\color{blue}\bfseries,
  identifierstyle=\color{black},
  sensitive=false,
  comment=[l]{//},
  commentstyle=\color{purple}\ttfamily,
}

\title{Foundations of the GraphAlg Language}
\date{}
\author{
    Daan de Graaf \\ {\small \href{mailto:d.j.d.graaf@tue.nl}{d.j.a.d.graaf@tue.nl}}
    \and Robert Brijder \\ {\small \href{mailto:r.brijder@tue.nl}{r.brijder@tue.nl}}
    \and Nikolay Yakovets \\ {\small \href{mailto:hush@tue.nl}{hush@tue.nl}}
}

\begin{document}

\maketitle

\begin{abstract}
    The \graphalg{} domain-specific language for graph algorithms enables user-defined algorithms in graph databases.
    In this work we show how \graphalg{} is built on top of the formal $\ML$ language for matrix manipulation.
    Starting from $\ML$, we describe the extensions to $\ML$ and the syntactic sugar needed to derive \graphalg{}.
    Furthermore, we prove that any \graphalg{} program can be simulated in an extension of for-$\ML$ that supports simultaneous induction.
\end{abstract}

\section{Introduction}
Graph algorithms are key components in many analytical workloads over graph data.
Despite graph databases being the natural home for such data, typical graph query languages such as Cypher~\cite{francis_cypher_2018} or GQL~\cite{iso_information_2024} are unable to express algorithms such as PageRank or Weakly Connected Components, forcing users to export their graph data and process it in external tools.
To bring user-defined algorithm support to graph databases, the \graphalg{} language has been proposed~\cite{graaf_algorithm_2026}.
\graphalg{} is a domain-specific language for writing graph algorithms in the language of linear algebra.
It is specifically designed to be embedded in databases: the \graphalg{} compiler~\cite{de_graaf_wildarchgraphalg_2026} transforms programs into an extended relational algebra that can be executed by an existing relational query processing engine with minimal changes.

The design of \graphalg{} is based on $\ML$~\cite{brijder_expressive_2019}, a formal language for matrix manipulation.
The $\ML$ language follows a principled and minimal design suitable for high-level optimization, with an established connection to relational algebra~\cite{brijder_matrices_2022}, making it an ideal starting point for \graphalg{}.
In this work we describe and motivate a number of extensions to $\ML$ for expressing a diverse set of graph algorithms, each adding a small feature to the previous extension, to eventually arrive at the \graphalg{} language.

By formalizing the connection between \graphalg{} and $\ML$, results from existing studies of $\ML$ can be leveraged for \graphalg{}.
For example, $\ML$ has been shown to correspond to a natural fragment of relational algebra~\cite{brijder_matrices_2022}.
We establish in this work that \graphalg{} programs can be simulated in an extension of for-$\ML$ that supports simultaneous induction (\cref{cor:sifor-eq}), suggesting that loop-free \graphalg{} translates to relational algebra.
Indeed, we implement this exact translation in the \graphalg{} compiler.

\paragraph*{Simultaneous Induction.}
Many graph algorithms require iteration, which is the primary reason that Cypher and GQL (language without loops or recursion) are insufficient.
While iteration support in $\ML$ has been previously studied~\cite{geerts_matrix_2021,geerts_expressive_2021}, the proposed extensions only allow a single matrix value to be carried over between loop iterations.
Various graph algorithms, however, are naturally expressed using multiple state variables carried over between loop iterations, also called \emph{simultaneous induction}.
Our proposed extension sifor-$\ML$, based on for-$\ML$~\cite{geerts_expressive_2021}, adds support for simultaneous induction to $\ML$.

\paragraph*{Leader Election and Loop Decomposition.}
Another important problem in graph analytics is the selection of a representative vertex from a set, also called \emph{leader election}.
We observe that while $\ML$ has no specific leader election operation, it can be simulated in sifor-$\ML$ using the powerful loop construct.
Conversely, if we do add a leader election primitive to the language, then the loop construct is unnecessarily powerful.
We propose an extension that decomposes the expressive power of the sifor-$\ML$ loop construct into two orthogonal constructs.
The resulting dec-$\ML$ language has simpler individual operations and can naturally express leader election, yet it has the same expressive power as sifor-$\ML$.

\paragraph*{Multiple Semirings.}
The next issue we face is that $\ML$ expressions operate over a fixed semiring, while the natural semiring to use depends on the specific algorithm.
Complex algorithms can even use multiple semirings at different stages, requiring casting of values between different semirings.
We propose muse-$\ML$, an extension of dec-$\ML$ that allows using multiple semirings within an expression.
We find that, perhaps counterintuitively, this extension does not increase the expressive power of the language: muse-$\ML$ and dec-$\ML$ (and therefore also sifor-$\ML$) have equal expressive power.

\paragraph*{Executable Specification.}
The $\ML$ language and its extensions mentioned above are all formal languages:
some details of the languages are irrelevant to theoretical analysis, and are therefore only loosely defined.
In particular, $\ML$ does not specify the set $\Omega$ of allowed pointwise functions, or which semirings may be used.
To create an executable specification of muse-$\ML$, we fix a particular $\Omega$ through a grammar, and define a set of semirings that is broad enough to support many graph algorithms.
The resulting language is \graphalg{} Core, the small subset of \graphalg{} that is used in the \graphalg{} compiler for optimization.

\paragraph*{Syntactic Sugar.}
While a small core subset is indispensable for compiler analysis and optimization, it can be cumbersome to write as an end user.
The `full' \graphalg{} language extends \graphalg{} Core with additional syntactic sugar for patterns commonly used in graph algorithms, without changing the expressive power.
We also adopt an imperative programming style with functions and statements that resembles Python, which we anticipate will be more familiar to users than the purely functional nature of \graphalg{} Core.

\bigskip
\noindent
The remainder of this paper is organized as follows.
In \cref{sec:prelim} we recall the definition of $\ML$, as well as the extension for-$\ML$ that adds iteration in the form of a loop construct.
\cref{sec:sifor-ml} defines sifor-$\ML$, which adds simultaneous induction.
In \cref{sec:dec-ml} we decompose this loop construct to obtain dec-$\ML$.
Next, \cref{sec:muse-ml} proposes muse-$\ML$, allowing multiple semirings to be used within a single expression.
We also show that muse-$\ML$ and dec-$\ML$ have equal expressive power.
In \cref{sec:core-lang} we introduce \graphalg{} Core, fixing $\Omega$ and the set of supported semirings.
Finally, in \cref{sec:full-lang} we give the formal definition of \graphalg{} and describe the procedure for converting programs into \graphalg{} Core expressions.

\section{Preliminaries}\label{sec:prelim}
We start by recalling $\ML$~\cite{brijder_expressive_2019}, followed by its extension for-$\ML$~\cite{geerts_expressive_2021}.

\subsection{MATLANG}
$\ML$ is a formal language for matrix manipulation using common matrix operations and linear algebra.
Previous work on $\ML$ has studied its expressive power~\cite{brijder_expressive_2019} and connection to relational algebra~\cite{brijder_matrices_2022}.

We assume a countably infinite set of matrix variables $\Mvar$ that serve as the inputs to $\ML$ expressions.
We also assume a repertoire $\Omega$ of functions $f : R^k \to R$, where $R$ is a semiring.
While we make no assumptions about $R$, previous works have used the complex numbers~\cite{brijder_expressive_2019} or the real numbers~\cite{geerts_expressive_2021} for $R$.
The syntax of $\ML$ expressions is defined by the following grammar:
\begin{align*}
    e & ::= \M \in \Mvar                      &  & \text{(matrix variable)}             \\
      & \mid\quad \transp{e}                  &  & \text{(transpose)}                   \\
      & \mid\quad \one(e)                     &  & \text{(one-vector)}                  \\
      & \mid\quad \diag(e)                    &  & \text{(diagonalization of a vector)} \\
      & \mid\quad e_1 \cdot e_2               &  & \text{(matrix multiplication)}       \\
      & \mid\quad \Apply[f](e_1, \ldots, e_k) &  & \text{(pointwise
        application, $f \in \Omega$)}
\end{align*}

The definition of $\ML$ in \cite{brijder_expressive_2019} also includes a let binding expression $\LetIn{\M}{e_1}{e_2}$ for convenience.
As noted in \cite{brijder_expressive_2019}, this operation can be omitted without affecting the expressive power of the language, so we omit it from our definition.

\subsubsection{Type System}
We now recall the type system for $\ML$ from \cite{brijder_expressive_2019}.
We assume a sufficient supply of \emph{size symbols}, denoted $\alpha,\beta,\gamma$, that represent the number of rows or columns of a matrix (e.g., $\alpha \times \beta$).
Together with an explicit $1$, we can express column vectors ($\alpha \times 1$), row vectors ($1 \times \beta$) and scalars ($1 \times 1$).
A \emph{size term} is either a size symbol or an explicit $1$.
Formally, a \emph{type} is an ordered pair of size terms $s_1$ and $s_2$, denoted $s_1 \times s_2$.
A \emph{schema} $\Sch$ is a function defined over a finite set of matrix variables $\Mnam \subset \Mvar$ that assigns a type to each element of $\Mnam$.
By $\Sch[M := \alpha \times \beta]$ we denote the schema obtained from $\Sch$ by mapping $M$ to $\alpha \times \beta$.
The type of an expression $e$ with respect to a schema $\Sch$, if it exists, is denoted by $\ttypes(e)$.
It is inductively defined as follows:

\begin{itemize}
    \item $\ttypes(M):= \Sch(M)$, for a matrix variable $M\in\Mnam$;
    \item $\ttypes(\transp{e}):= \beta \times \alpha$ if $\ttypes(e)=\alpha \times \beta$;
    \item $\ttypes(\one(e)):= \alpha \times 1$ if $\ttypes(e)=\alpha \times \beta$;
    \item $\ttypes(\diag(e)):= \alpha \times \alpha$, if $\ttypes(e)=\alpha \times 1$;
    \item $\ttypes(e_1 \cdot e_2):= \alpha \times \gamma$ if  $\ttypes(e_1)=\alpha \times \beta$, and $\ttypes(e_2)=\beta \times \gamma$;
    \item $\ttypes(\Apply[f](e_1,\ldots ,e_k)):= \alpha \times \beta$, if $\ttypes(e_i) = \alpha \times \beta$ for all $i \in \{1,\ldots,k \}$ and $f \in \Omega$.
\end{itemize}

We call an expression $e$ \textit{well-typed} according to the schema $\Sch$ if $\ttypes(e)$ exists.

\subsubsection{Semantics}
We denote the dimensions of a matrix $A$ by $\dim(A) = m \times n$.
An \emph{instance} $\I$ over a schema $\Sch$ is a pair of functions $\I = (\Idim, \Imat)$,
where $\Idim$ maps size terms to matrix dimensions, and $\Imat$ assigns a concrete matrix to each matrix variable $M \in \Mnam$,
such that $\dim(\Imat(M)) = \Idim(\alpha) \times \Idim(\beta)$ if $\ttypes(M) = \alpha \times \beta$.
We use the notation $\I(M)$ as a shorthand for $\Imat(M)$.
By $\I[M := A]$ we denote the instance obtained from $\I$ by mapping matrix variable $M$ to value $A$.
We write $(\I[M_1 := A])[M_2 := B]$ as $\I[M_1 := A, M_2 := B]$.
The \emph{evaluation} of an expression $e$ with respect to an instance $\I$, denoted by $\sem{e}{\I}$, is inductively defined as:

\begin{itemize}
    \item $\sem{M}{\I} := \I(M)$, for $M\in \Mnam$;
    \item $\sem{e^T}{\I} := \sem{e}{\I}^T$, which is the transpose of $\sem{e}{\I}$;
    \item $\sem{\one(e)}{\I}$ is an $m\times 1$ vector with value $1$ at all positions, where $\sem{e}{\I}$ is an $m \times n$ matrix;
    \item $\sem{\diag(e)}{\I}$ is a matrix with the vector $\sem{e}{\I}$ on its main diagonal, and zero in every other position;
    \item $\sem{e_1\cdot e_2}{\I} := \sem{e_1}{\I} \cdot \sem{e_2}{\I}$;
    \item $\sem{\Apply[f](e_1,\ldots ,e_k)}{\I}$ is a matrix $A$ of the same dimensions as $\sem{e_1}{\I}$, and where $A_{ij}$ has the value $f(\sem{e_1}{\I}_{ij},\ldots ,\sem{e_k}{\I}_{ij})$
\end{itemize}

\subsection{For-MATLANG}
The language for-$\ML$~\cite{geerts_expressive_2021} extends $\ML$ by a construct that we call the \emph{canonical for loop}.
Its syntax is as follows:

\medskip
\begin{tabular}{lc}
    $\ffor{v}{X}{e}$ & (canonical for loop, with $v, X \in \Mvar$)
\end{tabular}
\medskip

The associated typing rule is defined as:

\medskip
\begin{tabular}{lc}
    $\ttypes(\ffor{v}{X}{e}) := \ttypes(e)$, if $\ttypes(e) = \ttypes(X)$ and \\
    $\ttypes(v) = \gamma \times 1$
\end{tabular}
\medskip

Let $b_i^n$ be the vector of dimension $n$ with value $1$ at position $i$ and zero at all other positions.
Such a vector is called \emph{canonical}.
Assuming $\I(v)$ is an $n \times 1$ vector, $\sem{\ffor{v}{X}{e}}{\I}$ is defined iteratively as follows:

\begin{itemize}
    \item Let $A_0$ be the zero matrix of the same size as $\I(X)$.
    \item For $i \in \{ 1,\ldots,n \}$, $A_i:= \sem{e}{\I[v := b_i^n, X:= A_{i-1}]}$.
    \item Finally, $\sem{\ffor{v}{X}{e}}{\I}:= A_{n}$.
\end{itemize}

Intuitively, this loop construct repeatedly updates state variable $X$ by evaluating $e$, where $e$ is allowed to refer to the previous state of the loop as $X$.
Additionally, $e$ may refer to the \emph{canonical vector} $v$.
In the first iteration, $v$ takes on value $b_1^n$, in the second iteration $b_2^n$, etc.

We can define a variant of the above loop construct that initializes $X$ to an arbitrary expression $e_0$ rather than $A_0$.
This variant, denoted $\ffor{v}{X := e_0}{e}$, has the same expressive power as the zero-initialized loop construct\cite{geerts_expressive_2021}.

\section{Adding Simultaneous Induction}\label{sec:sifor-ml}
Many algorithms are naturally expressed using loops that update multiple state variables, also known as \emph{simultaneous induction}.
Simultaneous induction is usually supported in programming languages and in certain query languages~\cite{chandra_programming_1981}.
For example, consider that a typical breadth first search algorithm requires keeping three pieces of state:
\begin{enumerate}
    \item A queue of vertices to be visited;
    \item The vertices that have been visited so far; and
    \item The depth at which each vertex was first visited (or alternatively, a pointer to the parent vertex).
\end{enumerate}

The language for-$\ML$, however, only allows a single matrix to be carried from one iteration of a loop to the next.
The lack of simultaneous induction makes it difficult, if not impossible, to express various commonly used algorithms in for-$\ML$.
We discuss the implications for the expressive power of for-$\ML$ in detail in \cref{sec:si-proof}.

We propose an extension of for-$\ML$, named `sifor-$\ML$', that supports simultaneous induction.
The sifor-$\ML$ language is obtained from for-$\ML$, by modifying the loop definition to allow multiple loop variables:

\medskip
\begin{tabular}{lc}
    $\texttt{for}\, v,\{ M_1 := e_1', \ldots, M_m := e_m' \}(e_1, \ldots, e_m)$
\end{tabular}
\medskip

The typing rule is modified accordingly:

\medskip
\begin{tabular}{lc}
    $\ttypes(\texttt{for}\, v,\{ M_1 := e_1', \ldots, M_m := e_m' \}(e_d, e_1, \ldots, e_m)) := \tau_1$, \\
    if $\tau_i := \ttypes(e_i) = \ttype_{\Sch'}(e_i')$ for all $i \in \{ 1,\ldots,m \}$,                 \\
    where $\Sch' := \Sch[M_1 := \tau_1,\ldots,M_m := \tau_m]$
    and $\ttypes(v) = \gamma \times 1$.
\end{tabular}
\medskip

Assuming that $\I(v)$ is an $n \times 1$ vector, the semantics of the loop are defined iteratively as follows:

\begin{itemize}
    \item Let $\I_0 := \I[v := b^{n}_1, M_1 := \sem{e_1}{\I}, \ldots, M_m := \sem{e_m}{\I}]$
    \item For $i \in \{ 1,\ldots, n \}$, let $\I_i := \I_{i-1}[v := b^{n}_{i+1}, M_1 := \sem{e_1'}{\I_{i-1}}, \ldots, M_m := \sem{e_m'}{\I_{i-1}}]$.
    \item Finally, $\sem{\texttt{for}\, v,\{ M_1 := e_1', \ldots, M_m := e_m' \}(e_1, \ldots, e_m)}{\I}:= \I_n(M_1)$.
\end{itemize}

The loop body defines a number of loop variables $m \geq 1$ that each consist of a \emph{binding variable} $M_i$ and a \emph{body expression} $e_i'$.
Intuitively, before entering the loop, the binding variables are initialized as $M_i := e_i$.
Evaluation of the loop is similar to for-$\ML$, each time updating the binding variables with the results of the body expressions.
The loop produces a single result matrix, which is the final value of $M_1$.

The sifor-$\ML$ language subsumes for-$\ML$.

\begin{lemma}\label{lem:for-to-sifor}
    The loop construct of for-$\ML$ can be simulated using the loop construct of sifor-$\ML$.
\end{lemma}

The above lemma follows directly from the following rewrite rule:

\[
    \ffor{v}{X := e_0}{e} \Rightarrow \texttt{for}\, v,\{ X := e \}(e_0)
\]

\subsection{Does sifor-$\ML$ Strictly Subsume for-$\ML$?}\label{sec:si-proof}
Adding simultaneous induction to a language does not always increase its expressive power.
In first-order logic, for example, simultaneous induction can be simulated with induction on a single predicate~\cite[Chapter 8.2]{ebbinghaus_finite_1995}.
In this subsection we provide evidence to support the hypothesis that sifor-$\ML$ indeed has greater expressive power than for-$\ML$.
We consider possible strategies for simulating simultaneous induction with a single loop variable, and show that none of them are applicable to for-$\ML$.
We use the following simple recurrence as an example, which assumes two arbitrary input matrices $A$ and $B$:

\begin{align*}
    C_0 & = A                     \\
    C_1 & = B                     \\
    C_i & = C_{i-2} \cdot C_{i-1} \\
\end{align*}

Encoding this recurrence in for-$\ML$ would require the simulation of simultaneous induction, because the value of $C_i$ relies on the value of \emph{two} previous loop states $C_{i-2}$ and $C_{i-1}$.
For such a simulation to work, $C_{i-2}$ and $C_{i-1}$ must be encoded into a single matrix without loss of information.
Below we consider different strategies for combining two matrices $A$ and $B$ into one, and show that none of them can be applied to for-$\ML$.

\paragraph*{Additional Attributes.}
In relational algebra, relations can have arbitrary arity.
For the recurrence above, matrices $A$ and $B$ can be stored in relations with attributes $(r, c, a)$ and $(r, c, b)$, respectively.
To combine the two, we can simply create a table with four attributes $(r, c, a, b)$.
In relational algebra, simultaneous induction therefore does not increase expressivity.
Given that for-$\ML$ has the same expressivity as a particular fragment of relational algebra~\cite{brijder_matrices_2022}, one might assume that the same holds in for-$\ML$.
Crucially, however, in for-$\ML$ the matrix type limits us to arity 3~\cite{brijder_matrices_2022}, so the strategy of using additional attributes does not carry over to for-$\ML$.

\paragraph*{Matrix Concatenation.}
One might wonder whether in for-$\ML$ two $m \times n$ matrices can be concatenated to form a single $2m \times n$ (or $m \times 2n$) matrix.
However, given the type rules of for-$\ML$ it is apparent that such a matrix cannot be constructed.
Consider a schema $\Sch$ with a mapping for two matrix variables $\Sch[A := \alpha \times \beta]$ and $\Sch[B := \gamma \times \delta]$.
Then the only possible values of $\ttypes(e)$ for arbitrary $e$ are combinations of $\alpha, \beta, \gamma, \delta$ and 1.
Matrices must be constructed in terms of the dimensions of the inputs, and therefore concatenation of arbitrary matrices is impossible.

\paragraph*{Pairing.}
Two matrices $A$ and $B$ can be encoded in a single matrix using a pairing function $\pi : R \times R \to R$ that maps each pair of elements $(A_{ij},B_{ij})$ into a distinct value.
If $\pi$ is an injective mapping, then applying $\pi$ to matrices $A$ and $B$ gives a single matrix $C$ from which both $A$ and $B$ can be recovered using the retraction of $\pi$.
For finite semirings, however, it is easy to show that $\pi$ cannot be injective.
Let $R$ be a finite semiring with cardinality $|R|$.
Then there are $|R|^2$ distinct values in the domain of $\pi$, but only $|R|$ in the co-domain, which is insufficient to map each pair $(A_{ij},B_{ij})$ to a distinct value.
The implications of this are especially damning for a software implementation of for-$\ML$, which would need to operate over fixed-size data types for performance.

\bigskip
\noindent
None of the strategies described above succeed in simulating simultaneous induction in for-$\ML$.
This supports the hypothesis that because of its support for simultaneous induction, the expressive power of sifor-$\ML$ is indeed greater than for-$\ML$.

\section{Decomposing the Loop Construct}\label{sec:dec-ml}
The loop construct in sifor-$\ML$ can be replaced by two natural constructs:
A simpler loop definition without canonical vectors, as has been previously proposed in the context of $\ML$~\cite{geerts_matrix_2021} and query languages in general~\cite{chandra_programming_1981},
and an additional operation $\pickAny$, a type of aggregation in relational algebra, that is independently useful for tasks such as leader election.
This decomposition also has further advantages, which we discuss at length in \cref{sec:dec-advantages}.

Before giving the loop definition, we first define the typing rule and semantics of the $\pickAny$ construct:

\medskip
\begin{tabular}{lc}
    $\ttypes(\pickAny(e)) := \ttypes(e)$
\end{tabular}
\medskip

The expression $\pickAny(e)$ keeps the non-zero element with minimal column index of each row in $e$, setting all other entries to zero:

\medskip
\begin{tabular}{lc}
    $\sem{\pickAny(e)}{\I}$ is a matrix $A$ of the same size as $\sem{e}{\I}$, and where $A_{ij}$ \\
    is $\sem{e}{\I}_{ij}$ if $\sem{e}{\I}_{ik} = 0$ for all $k < j$, and $0$ otherwise.
\end{tabular}
\medskip

For the updated loop construct, we drop the canonical vector $v$, and instead use an input expression $e_d$ to determine the number of iterations that the loop should run:

\medskip
\begin{tabular}{lc}
    $\texttt{for}\, \{ M_1 := e_1', \ldots, M_m := e_m' \}(e_d, e_1, \ldots, e_m)$
\end{tabular}
\medskip

The updated typing rule is as follows:

\medskip
\begin{tabular}{lc}
    $\ttypes(\texttt{for}\, \{ M_1 := e_1', \ldots, M_m := e_m' \}(e_d, e_1, \ldots, e_m)) := \tau_1$, \\
    if $\tau_i := \ttypes(e_i) = \ttype_{\Sch'}(e_i')$ for all $i \in \{ 1,\ldots,m \}$,               \\
    where $\Sch' := \Sch[M_1 := \tau_1,\ldots,M_m := \tau_m]$
    and $\ttypes(e_d) = \gamma \times 1$.
\end{tabular}
\medskip

Assume that $\sem{e_d}{\I}$ is an $n \times 1$ vector.
The semantics are almost identical, except that $\I_i$ does not include a matrix variable mapping to a canonical vector:
\begin{itemize}
    \item Let $\I_0 := \I[M_1 := \sem{e_1}{\I}, \ldots, M_m := \sem{e_m}{\I}]$
    \item For $i = \{ 1,\ldots n \}$, let $\I_i := \I_{i-1}[M_1 := \sem{e_1'}{\I_{i-1}}, \ldots, M_m := \sem{e_m'}{\I_{i-1}}]$.
    \item Finally, $\sem{\texttt{for}\, \{ M_1 := e_1', \ldots, M_m := e_m' \}(e_d, e_1, \ldots, e_m)}{\I}:= \I_n(M_1)$.
\end{itemize}

We denote the language obtained from sifor-$\ML$ by simplifying the loop construct and adding $\pickAny$ as dec-$\ML$.
It turns out that these changes do not affect the expressive power of the language.

\begin{theorem}\label{thm:sifor-eq-dec}
    Languages sifor-$\ML$ and dec-$\ML$ have equal expressive power.
\end{theorem}

\begin{proof}
    We give constructions showing that:
    (1) The loop construct of sifor-$\ML$ can be expressed in dec-$\ML$;
    (2) The loop construct of dec-$\ML$ can be expressed in sifor-$\ML$; and
    (3) The $\pickAny$ construct of dec-$\ML$ is expressible in sifor-$\ML$.

    \paragraph*{Simulating sifor-$\ML$ loops in dec-$\ML$.}
    The rewrite rule below shows how loops with canonical vectors as they appear in sifor-$\ML$ can be expressed in dec-$\ML$.

    \begin{tabbing}
        L0\=L1\=L2\=L3\=L4 \kill
        $\texttt{for}\, v,\{ M_1 := e_1', \ldots, M_m := e_m' \}(e_1, \ldots, e_m) \Rightarrow$ \\
        $\texttt{for}\, \{$ \\
        \>$M_1 := e_1', \ldots, M_m := e_m'$\\
        \>$v := \transp{\pickAny(V - \transp{v})}$ \\
        \>$V := V - \transp{v}$ \\
        $\}($\\
        \>$v,$\\
        \>$e_1, \ldots, e_m,$ \\
        \>$\transp{\pickAny(\transp{\one(v)})},$\\
        \>$\transp{\one(v)})$
    \end{tabbing}

    The rewrite introduces two additional loop variables $v$ and $V$, where
    $v$ stores the current canonical vector, while $V$ tracks which canonical vectors have been observed in previous iterations (corresponding to the zero values).
    Variable $V$ must not appear in $\{ M_1,\ldots,M_m \}$.
    Because $\Mvar$ is countably infinite, we can always find such a $V$.
    The input expressions to the loop initialize these as $v := \transp{\pickAny(\transp{\one(v)})} = \transp{[1\ 0 \cdots 0]}$ and $V := \transp{\one(v)} = [1 \cdots 1]$.
    After the first iteration, the variables are updated to $v := \transp{\pickAny(V - \transp{v})} = \transp{[0 \ 1\ 0 \cdots 0]}$ and $V := V - \transp{v} = [0\ 1 \cdots 1]$, and so on.

    \paragraph*{Simulating dec-$\ML$ loops in sifor-$\ML$.}
    dec-$\ML$ loops are a simplification of sifor-$\ML$ loops, which makes rewriting them to sifor-$\ML$ trivial:

    \begin{align*}
        \texttt{for}\ \ \  & \{ M_1 := e_1', \ldots, M_m := e_m' \}(v, e_1, \ldots, e_m)
        \Rightarrow                                                                         \\
        \texttt{for}\, v,  & \{ M_1 := e_1', \ldots, M_m := e_m' \}(\ \ \ e_1, \ldots, e_m)
    \end{align*}

    \paragraph*{Simulating $\pickAny$ in sifor-$\ML$.}
    $\pickAny(A)$ can be simulated using the for loop with canonical vectors of for-$\ML$ (and by extension also in sifor-$\ML$).
    Let $f \in \Omega$ be the pointwise function defined as:

    \[
        f(y, d, p) = \begin{cases}
            p & d = 0    \\
            y & d \neq 0
        \end{cases}
    \]

    Then $\pickAny(A)$ can be simulated as:

    \begin{tabbing}
        L0\=L1\=L2\=L3\=L4 \kill
        $\pickAny(A) \Rightarrow$ \\
        $\SyntaxStyle{let}\ v = \one(A) \ \SyntaxStyle{in}$ \\
        $\SyntaxStyle{let}\ w = \one(\transp{A}) \ \SyntaxStyle{in}$ \\
        $\SyntaxStyle{let}\ X = A \ \SyntaxStyle{in}$ \\
        $\SyntaxStyle{let}\ Y = \transp{w} \ \SyntaxStyle{in}$ \\
        $\SyntaxStyle{let}\ B = w \cdot \transp{w} \ \SyntaxStyle{in}$ \\
        $\texttt{for}\, v,X \texttt{.}\,$\\
        \>$\SyntaxStyle{let}\ R = \transp{v} \cdot A \ \SyntaxStyle{in}$ \\
        \>$\SyntaxStyle{let}\ F = \texttt{for}\, w,Y \texttt{.}\,$\\
        \>\>$\SyntaxStyle{let}\ D=Y \cdot B\ \SyntaxStyle{in}$ \\
        \>\>$\SyntaxStyle{let}\ P=R \cdot \diag(w)\ \SyntaxStyle{in}$ \\
        \>\>$\Apply[f](Y, D, P)$ \\
        \>$X + v \cdot F$
    \end{tabbing}

    Initial values are assigned to $v$, $w$, $X$, $Y$ only to define their dimensions, as this is needed for type checking the \texttt{for} loops.
    The outer loop $\texttt{for}\, v,X$ iterates over the input rows, while $\texttt{for}\, w,Y$ iterates over the columns.
    $Y$ contains the (partial) result for input row $R$.
    In the inner loop, we first detect ($D$) if we already have a value for the current row.
    This uses the matrix $B$, which broadcasts any nonzero value in $Y$ to all positions in $D$.
    Then, we take the proposed ($P$) value at position $(v, w)$ in $A$, setting all other values to zero.
    Finally, the $\Apply$ either picks the proposed value $p$ if there is no value yet for the current row, or the existing value $y$ otherwise.

    \paragraph*{Other Language Constructs.}
    All other language constructs from sifor-$\ML$ and dec-$\ML$ exist in both languages with the same semantics.
    Combining this with the rewrite rules derived above, we conclude that any well-typed sifor-$\ML$ expression can be expressed in dec-$\ML$ and vice versa.
\end{proof}

\subsection{Advantages of Loop Decomposition}\label{sec:dec-advantages}
A first advantage of the decomposition of the sifor-$\ML$ loop construct is that it allows studying the individual contributions of the loop construct and of $\pickAny$.
Both the loop construct and the $\pickAny$ operation are known from the literature.
The fragment of for-$\ML$ that omits canonical vectors was first proposed in \cite{geerts_matrix_2021}, and a similar loop construct has been defined for relational algebra~\cite{chandra_programming_1981}.
The $\pickAny$ construct can be naturally expressed in an extension of relational algebra that includes aggregation.
Let $A$ be an adjacency matrix, and $E$ the relation with attributes \texttt{(row, col, value)} that describes matrix $A$.
Assuming an aggregate function \texttt{ARGMIN(a, b)} that returns the value \texttt{a} for which \texttt{b} is minimal, the expression $\pickAny(A)$ corresponds to the following SQL query:
\begin{verbatim}
SELECT row, MIN(col), ARGMIN(value, col) 
FROM E
WHERE value <> 0
GROUP BY row
\end{verbatim}
Considering that the simulation of $\pickAny$ in sifor-$\ML$ requires complex constructs such as nested for loops, a translation from a simulated $\pickAny$ expression in sifor-$\ML$ to efficient relational algebra appears much more challenging.

The $\pickAny$ operation has uses that go beyond the construction of canonical vectors.
Consider the problem of finding weakly connected components in an undirected graph $G$.
Recall that the weakly connected components of $G$ partition the graph into disjoint subgraphs where the vertices within each subgraph are connected.
If we pick a \emph{representative} vertex for each connected component, a solution can be encoded as a function $L: V \to V$ that maps each vertex $v$ in the graph to the representative vertex of the connected component of $v$, where $V$ is the set of vertices of $G$.
Assuming an adjacency matrix $A$ that describes $G$, an algorithm for finding the weakly connected components is concisely expressed in dec-$\ML$ as:

\[
    \mathsf{WCC}(A) = \texttt{for}\ \{ X := \pickAny(X + (A \cdot X)) \}(\diag(\one(A)))
\]

Where $L(u) = v$ if $\mathsf{WCC}(A)_{uv} = 1$.
The intuition behind this algorithm is as follows: The initial value of $X$ is an identity matrix (written as $\diag(\one(A))$), representing the initial assumption that every vertex is in a distinct connected component.
Inside the loop, $A \cdot X$ is a one-hop graph traversal step that for each vertex produces the representative vertices of its direct neighbors.
The expression $\pickAny(X + (A \cdot X))$ states that for every vertex (i.e., every row), we pick the lowest label (i.e., the lowest column position) that appears among the immediate neighborhood ($A \cdot X$) or the current label ($X$).
Performing this repeatedly in a loop, the algorithm converges to a labeling $X$ in at most $|V|$ steps.
The definition of $\mathsf{WCC}$ facilitates a straightforward translation into an extension of relational algebra that can be efficiently executed~\cite{graaf_algorithm_2026}.

\section{Semirings}\label{sec:muse-ml}
$\ML$ expressions operate over a fixed semiring, which in prior work has been the complex~\cite{brijder_expressive_2019} or real numbers~\cite{geerts_expressive_2021}.
If instead we also allow other semirings to be used, thereby changing the semantics of matrix multiplication, we can elegantly express various useful algorithms.

\begin{example}
    Assuming the boolean semiring, the following algorithm computes the vertices reachable from a start vertex $s$ in a graph with adjacency matrix $A$,
    where $S$ is a column vector with value $1$ at $S_s$ and $0$ at all other positions.

    \[
        \mathsf{reach}(S, A) = {\normalfont \texttt{for}}\ \{ R := R + \transp{(\transp{v} \cdot A)}\}(S, S)
    \]

    Vector $R$ has value $1$ at $R_i$ if vertex $i$ is reachable from $s$.
    The algorithm can be modified to compute the shortest paths from $s$ to all other vertices in the graph, simply by switching from the boolean to the tropical min-plus semiring~\cite{kepner_graph_2011}.
    In this case, $R_i$ is the shortest distance from $s$ to $i$, or $\infty$ if there is no path between $s$ and $i$.
    Adjacency matrix $A$ contains the edge weights (with $\infty$ at $A_{uv}$ if there is no edge from $u$ to $v$).
    Vector $S$ has value $0$ at $S_s$ and $\infty$ at all other positions.
\end{example}

The example above motivates an extension of dec-$\ML$ that allows using different semirings for different algorithms.
Moreover, it is not difficult to imagine algorithms that use multiple semirings.
The label propagation algorithm we implement in \cite{graaf_algorithm_2026} defines its inputs and outputs over the integer semiring, but internally performs some operations over a tropical semiring.
We show one such case in the example below.

\begin{example}
    Consider a vector $V$ in the integer semiring from which we want to extract the maximum value.
    This is difficult to express in $\ML$ and its extensions discussed so far.
    Summing over all values of $V$, however, is trivial because it uses the natural addition operator of the integer semiring:

    \[
        \mathsf{sum}(V) = \transp{\one(V)} \cdot V
    \]

    If we permit \emph{casting} values to different semirings, we can leverage that to implement the $\mathsf{max}$ function.
    Let $\cast_{R \to T}$ be a family of functions that change the semiring $R$ of the input value to semiring $T$.
    Then the maximum value of a vector $V$ can be expressed as:

    \[
        \mathsf{max}(V) = \mathsf{sum}(\cast_{\ringInt \to \ringTropMaxInt}(V))
    \]

    Vector $V$ is cast from the integer to the tropical max-plus semiring, which uses the $\max$ operator for addition.
    Because of the semiring cast, the $\mathsf{sum}$ function computes the maximum instead of the sum.
\end{example}

\subsection{Adding Multiple Semiring Support to $\ML$}
Casting does not need to be added as a new language construct.
Instead, we generalize the $\Apply$ construct to allow pointwise functions with different semirings for inputs and outputs.
The hypothetical cast operation $\cast_{R \to T}(e)$ from the example above is then expressed as $\Apply[\cast_{R \to T}](e)$.
This approach does, however, require changing the data model to associate a semiring with each expression and matrix value.

Let muse-$\ML$ be the language obtained from dec-$\ML$ with the following modifications.
We assume a nonempty finite set $\texttt{SRing}$ of semirings.
In muse-$\ML$ a type is an ordered triple $s_1 \times s_2 \times r$, where $s_1$ and $s_2$ are size terms as in $\ML$, and $r \in \texttt{SRing}$ is a semiring.
To allow pointwise functions that cast between semirings, we define repertoire $\Omega$ to be a set of functions $f : R_1 \times \cdots \times R_k \to R_o$, where $R_1,\ldots,R_k$ and $R_o$ are all (potentially distinct) semirings from $\texttt{SRing}$.
The updated typing rules are presented below:

\begin{itemize}
    \item $\ttypes(M):= \Sch(M)$, for a matrix variable $M\in\Mnam$;
    \item $\ttypes(\transp{e}):= \beta \times \alpha \times R$ if $\ttypes(e)=\alpha \times \beta \times R$;
    \item $\ttypes(\one(e)):= \alpha \times 1 \times R$ if $\ttypes(e)=\alpha \times \beta \times R$;
    \item $\ttypes(\diag(e)):= \alpha \times \alpha \times R$, if $\ttypes(e)=\alpha \times 1 \times R$;
    \item $\ttypes(e_1 \cdot e_2):= \alpha \times \gamma \times R$ if  $\ttypes(e_1)=\alpha \times \beta \times R$, and $\ttypes(e_2)=\beta \times \gamma \times R$;
    \item $\ttypes(\Apply[f](e_1,\ldots ,e_k)):= \alpha \times \beta \times R_o$, if $\ttypes(e_i) = \alpha \times \beta \times R_i$ for all $i \in [1 \ldots k]$ and $f : R_1 \times \cdots \times R_k \to R_o$ is in $\Omega$;
    \item $\ttypes(\pickAny(e)) := \alpha \times \beta \times R$ if $\ttypes(e) = \alpha \times \beta \times R$;
    \item $\ttypes(\texttt{for}\, \{ M_1 := e_1', \ldots, M_m := e_m' \}(e_d, e_1, \ldots, e_m)) := \tau_1$,\\
          if $\tau_i := \ttypes(e_i) = \ttype_{\Sch'}(e_i')$ for all $i \in \{ 1,\ldots,m \}$,\\
          where $\Sch' := \Sch[M_1 := \tau_1,\ldots,M_m := \tau_m]$
          and $\ttypes(e_d) = \gamma \times 1 \times R_d$.
\end{itemize}

The semantics of muse-$\ML$ expressions are as expected.
For example, the property that $\sem{e_1 \cdot e_2}{\I} = \sem{e_1}{\I} \cdot \sem{e_2}{\I}$ is maintained.
The only difference is that where in dec-$\ML$ there is one semiring $R$ that defines the addition and multiplication operator (and the additive and multiplicative identities), in muse-$\ML$ we use the semiring associated with the input.

\subsection{Semiring Casting Does Not Add Expressivity}
Perhaps surprisingly, muse-$\ML$ has the same expressive power as dec-$\ML$.

Before proving this, we first define the following helper functions.
We use for-$\ML$ and sifor-$\ML$ style loops with canonical vectors where this aids readability, noting that we have previously shown in \cref{lem:for-to-sifor,thm:sifor-eq-dec} that such loops can be rewritten into dec-$\ML$.
Function $\mathsf{emax}(V)$ follows the definition given in \cite{geerts_expressive_2021}, returning the last canonical vector:

\begin{tabbing}
    L0\=L1\=L2\=L3\=L4 \kill
    $\mathsf{emax}(V) = $\\
    $\SyntaxStyle{let}\ v = V \ \SyntaxStyle{in}$ \\
    $\SyntaxStyle{let}\ X = V \ \SyntaxStyle{in}$ \\
    $\texttt{for}\, v,X \texttt{.}\, v$
\end{tabbing}

The converse $\mathsf{emin}$, which returns the first canonical vector, can be defined as $\mathsf{emin}(V) = \transp{\pickAny(\transp{\one(V)})}$.
We also take the definition of ordering matrix $S_{\leq}$ from \cite{geerts_expressive_2021}:

\[
    \begin{bmatrix}
        1      & 1      & \cdots & 1 \\
        0      & 1      & \cdots & 1 \\
        \vdots & \vdots & \ddots & 1 \\
        0      & 0      & \cdots & 1
    \end{bmatrix}
\]

This matrix is obtained by evaluating the following function over a column vector $V$:

\begin{tabbing}
    L0\=L1\=L2\=L3\=L4 \kill
    $S_{\leq}(V) = $\\
    $\SyntaxStyle{let}\ v = V \ \SyntaxStyle{in}$ \\
    $\SyntaxStyle{let}\ X = \diag(V) \ \SyntaxStyle{in}$ \\
    $(\texttt{for}\, v,X \texttt{.}$\\
    \>$X
        + (X \cdot \mathsf{emax}(V) + v) \cdot \transp{v}
        + v \cdot \transp{\mathsf{emax}(V)}
        ) - \one(V) \cdot \transp{\mathsf{emax}(V)}$
\end{tabbing}

Note that $S_{\leq}(V)$ evaluates to a square matrix whose the dimensions match the number of rows in $V$.
The final subtraction term $- \one(V) \cdot \transp{\mathsf{emax}(V)}$ is missing in \cite{geerts_expressive_2021}, but is required to avoid some elements of the final matrix having value $2$ rather than the expected $1$.

By subtracting the identity matrix, we obtain $S_{<}(V) = S_{\leq}(V) - \diag(\one(V))$.
Based on $S_{<}(V)$, we can define $\mathsf{rotate}(V)$:

\[
    \mathsf{rotate(V)} = \pickAny(S_{<}(V)) + \mathsf{emax}(V) \cdot \transp{\mathsf{emin}(V)}
\]

The expression $\mathsf{rotate}(V)$ evaluates to:

\[
    \begin{bmatrix}
        0      & 1      & 0      & \cdots & 0      \\
        \vdots & \ddots & \ddots & \ddots & \vdots \\
        \vdots &        & \ddots & \ddots & 0      \\
        0      &        &        & \ddots & 1      \\
        1      & 0      & \cdots & \cdots & 0      \\
    \end{bmatrix}
\]

With the same dimensions as $S_{\leq}(V)$.
It is a permutation matrix with the property that when it is multiplied with a vector, it cyclically rotates all entries by one position up.
For example, given $V = \transp{[1\ 2\ 3]}$, we have $\mathsf{rotate}(V) \cdot V = \transp{[2\ 3\ 1]}$.

Function $\mathsf{sum}[\underline{0}, \oplus](V)$ generalizes summation of (column) vector $V$ to a scalar value.
Parameters are the additive identity $\underline{0}$ and addition operator $\oplus$.

\begin{tabbing}
    L0\=L1\=L2\=L3\=L4 \kill
    $\mathsf{sum}[\underline{0}, \oplus](V) = $\\
    $\SyntaxStyle{let}\ R = \mathsf{rotate}(V) \ \SyntaxStyle{in}$ \\
    $\SyntaxStyle{let}\ X =  \texttt{for}\, \{$\\
    \>\>$X:= \Apply[\oplus](V, R \cdot X)$ \\
    \>$\}(V, \Apply[\underline{0}](V))$\\
    $\SyntaxStyle{in}\ \transp{\mathsf{emax}(V)} \cdot X$
\end{tabbing}

The expression $\mathsf{sum}[\underline{0}, \oplus](V)$ evaluates to $\underline{0} \oplus V_1 \oplus \cdots \oplus V_n$, where $V$ is an $n \times 1$ vector.
Repeated multiplication with permutation matrix $R$ followed by element-wise addition ensure that all entries of $X$ contain the sum of values in the vector, i.e., given $V = \transp{[1\ 2\ 3]}$ we have $X = \transp{[6\ 6\ 6]}$.
A final multiplication with $\transp{\mathsf{emax}(V)}$ extracts the last element of $X$ (though since all values are equal, $\transp{\mathsf{emin}(V)}$ would also work).

We are now ready to present the proof that muse-$\ML$ has the same expressive power as dec-$\ML$.
Recall that for any two sets $X$ and $Y$, if $X$ has cardinality no greater than $Y$ ($|X| \leq |Y|$), then there is an injective function $f$ from $X$ to $Y$.
The \emph{retraction} $g$ of $f$ is a function from $Y$ to $X$ such that for all $x \in X$ we have $g(f(x)) = x$.
Note that if $f$ is injective, then a retraction $g$ exists.

\begin{theorem}\label{thm:muse-eq-dec}
    Let $e$ be a well-typed muse-$\ML$ expression, and let $\texttt{SRing}$ be the set of semirings used by $e$.
    Let $T$ be a semiring in $\texttt{SRing}$ such that for all $R \in \texttt{SRing}$ we have $|R| \leq |T|$.
    Then there exists an expression in dec-$\ML$ defined over $T$ that is equivalent to $e$.
\end{theorem}

\begin{proof}
    Let $\mathsf{enc}_R : R \to T$ be an arbitrary injective function from $R$ to $T$, and $\mathsf{dec}_R : T \to R$ a retraction of $\mathsf{enc}_R$.
    Let $\Omega$ be the pointwise function repertoire used in $e$, and let $f : R_1 \times \cdots \times R_k \to R_o$ be in $\Omega$.
    Then we define:
    \[
        \encoded{f}(e_1, \ldots, e_k) = \mathsf{enc}_o(f(\mathsf{dec}_{R_1}(e_1), \ldots, \mathsf{dec}_{R_k}(e_k)))
    \]

    Intuitively, $\encoded{f}$ is an encoded version of $f$ operating in semiring $T$.
    Additionally, for any semiring $R \in \texttt{SRing}$, let $\oplus_R$ and $\otimes_R$ be the addition and multiplication operator of $R$, respectively.
    We define the following operators:
    \begin{align*}
        e_1 \encoded{\oplus_R} e_2 = \mathsf{enc}_R(\mathsf{dec}_R(e_1) \oplus_R \mathsf{dec}_R(e_2)) \\
        e_1 \encoded{\otimes_R} e_2 = \mathsf{enc}_R(\mathsf{dec}_R(e_1) \otimes_R \mathsf{dec}_R(e_2))
    \end{align*}

    Let $\encoded{\Omega}$ be the repertoire consisting of:
    \begin{itemize}
        \item $\encoded{f}$ for all functions $f \in \Omega$;
        \item $\encoded{\oplus_R}$ and $\encoded{\otimes_R}$ for all $R_R \in \texttt{SRing}$;
        \item A helper function $\mathsf{cond}$, defined as:
              \[
                  \mathsf{cond}(w, x, y, z) = \begin{cases}
                      y & w = x    \\
                      z & w \neq x
                  \end{cases}
              \]
        \item A binary operator $(-)$ for which the usual identities $0 - 0 = 1 - 1 = 0$ and $1 - 0 = 1$ hold.
              For other input values the behavior may be assumed as arbitrary.
    \end{itemize}

    Because all functions in $\encoded{\Omega}$ are defined over $T$, $\encoded{\Omega}$ is indeed a repertoire for dec-$\ML$ over $T$.

    We show that any muse-$\ML$ expression can be rewritten into an equivalent dec-$\ML$ expression over $T$ with the repertoire $\encoded{\Omega}$.
    The proof now proceeds by induction on the structure of expressions.
    We give rewrite rules of the form $\texttt{OP}(e) \Rightarrow \texttt{OP}'(\encoded{e})$ that describe how to rewrite muse-$\ML$ constructs into dec-$\ML$ under the induction hypothesis that inner muse-$\ML$ expressions ($e$ in this example) can be rewritten into \emph{semantically equivalent} dec-$\ML$ expressions ($\encoded{e}$),
    i.e., $\sem{e}{\I} = \mathsf{dec}_R(\sem{\encoded{e}}{\encoded{\I}})$, where $e$ is defined over semiring $R$, and instance $\encoded{\I}$ is defined as:

    \[
        \encoded{\I}(M) = \mathsf{enc}_R(A) \text{ if } \I(M) = A \text{ and } A \text{ is a matrix over semiring } R
    \]

    \textbf{Transpose, Diagonalization and Loops.}
    A first observation is that the semantics of $\transp{e}$, $\diag(e)$ and loops is independent of the semirings associated with the inner expressions, and hence the rewrite rules for these expressions are trivial:

    \begin{align*}
        \transp{e}                                                                  & \Rightarrow \transp{\encoded{e}}                \\
        \diag(e)                                                                    & \Rightarrow \diag(\encoded{e})                  \\
        \texttt{for}\ \{ M_1 := e_1', \ldots, M_m := e_m' \}(e_d, e_1, \ldots, e_m) & \Rightarrow                                     \\
        \texttt{for}\ \{ M_1 := \encoded{e_1'}, \ldots, M_m := \encoded{e_m'} \}(\encoded{e_d}, \encoded{e_1}, \ldots, \encoded{e_m}) \\
    \end{align*}

    \textbf{Pointwise Function Application.}
    By definition of $\encoded{f}$, the rewrite rule for pointwise function application is defined as:

    \[
        \Apply[f](e_1, \ldots, e_k) \Rightarrow \Apply[\encoded{f}](\encoded{e_1}, \ldots, \encoded{e_k})
    \]

    \textbf{One-Vector.}
    The expression $\one(e)$ must be rewritten because it would fill the output vector with the multiplicative identity as defined by $T$, rather than the multiplicative identity of the semiring $R$ (denoted $1_R$) over which $e$ is defined.
    To compensate we rewrite the expression as

    \[
        \one(e) \Rightarrow
        \Apply[\lambda x.\ \mathsf{enc}_R(1_R)](\one(\encoded{e}))
    \]

    \textbf{pickAny.}
    The semantics of $\pickAny(e)$ depend on the additive identity $0_R$ of the semiring $R$ that $e$ is defined over.
    Let $0_T$ and $1_T$ be the additive and multiplicative identities of semiring $T$, respectively.
    Then the rewrite rule is:

    \begin{tabbing}
        L0\=L1\=L2\=L3\=L4 \kill
        $\pickAny(e) \Rightarrow \Apply[\lambda v,p.\ \mathsf{cond}(p, 1_T, v, \mathsf{enc}_R(0_R))]($\\
        \>$\encoded{e},$\\
        \>$\pickAny(\Apply[\lambda v.\ \mathsf{cond}(v, \mathsf{enc}_R(0_R), 0_T, 1_T)](\encoded{e})))$
    \end{tabbing}

    Intuitively, the inner $\Apply$ maps zero (according to $R$) entries in $A$ to $0_T$, and all other values to $1_T$.
    Evaluating $\pickAny$ over this matrix returns a matrix with at most one entry per row with value $1_T$, at the first position where $A$ has a non-zero value.
    Finally, the outer $\Apply$ uses the result of $\pickAny$ as a mask to select which values from $A$ to preserve, setting the others to $0_R$.

    \textbf{Matrix Multiplication.}
    The rewrite rule for matrix multiplication is non-trivial because it requires simulating an addition and multiplication operator that may be very different from the ones associated with semiring $T$.
    We define a helper function $\mathsf{cellmul}$ to compute a single value corresponding to $C_{rc}$ of matrix multiplication $C = A \cdot B$ given row $A_r$ and column $B_c$, parameterized over the additive identity $\underline{0}$, addition operator $\oplus$ and multiplication operator $\otimes$:

    \[
        \mathsf{cellmul[\underline{0}, \oplus, \otimes]}(e_1, e_2) = \mathsf{sum}[\underline{0}, \oplus](\Apply[\otimes](\transp{e_1}, e_2))
    \]

    By iterating over the rows of $A$ and the columns of $B$ we can express arbitrary matrix multiplication, computing the entries one by one:

    \begin{tabbing}
        L0\=L1\=L2\=L3\=L4 \kill
        $\mathsf{rowmul}[\underline{0}, \oplus, \otimes](e_1, e_2) = $\\
        $\SyntaxStyle{let}\ v = \one(\transp{e_2}) \ \SyntaxStyle{in}$ \\
        $\texttt{for}\,v\ \{$\\
        \>$X := \SyntaxStyle{let}\ B_c = e_2 \cdot v \ \SyntaxStyle{in}$ \\
        \>\>$X + \mathsf{cellmul}[\underline{0}, \oplus, \otimes](e_1, B_c) \cdot \transp{v}$\\
        $\}(\Apply[0](e_1 \cdot e_2))$
    \end{tabbing}

    \begin{tabbing}
        L0\=L1\=L2\=L3\=L4 \kill
        $\mathsf{matmul}[\underline{0}, \oplus, \otimes](e_1, e_2) = $\\
        $\SyntaxStyle{let}\ v = \one(e_1) \ \SyntaxStyle{in}$ \\
        $\texttt{for}\,v\ \{$\\
        \>$X := \SyntaxStyle{let}\ A_r = \transp{v} \cdot e_1 \ \SyntaxStyle{in}$ \\
        \>\>$X + v \cdot \mathsf{rowmul}[\underline{0}, \oplus, \otimes](A_r, e_2)$\\
        $\}(\Apply[0](e_1 \cdot e_2))$
    \end{tabbing}

    The rewrite rule for matrix multiplication $e_1 \cdot e_2$ over semiring $R$ is now simply:

    \begin{tabbing}
        L0\=L1\=L2\=L3\=L4 \kill
        $e_1 \cdot e_2 \Rightarrow \mathsf{matmul}[\mathsf{enc}_R(0_R), \encoded{\oplus_R}, \encoded{\otimes_R}](\encoded{e_1}, \encoded{e_2})$
    \end{tabbing}

    The above rewrite rules cover all language constructs in muse-$\ML$, and therefore we conclude that for any muse-$\ML$ expression there exists an equivalent dec-$\ML$ expression.
\end{proof}

\section{GraphAlg Core}\label{sec:core-lang}
We are now ready to define an \emph{instantiation} of the muse-$\ML$ language (defined in \cref{sec:muse-ml}), called \graphalg{} Core.
\graphalg{} Core is the simplified, or \emph{desugared}, version of the \graphalg{} language defined in \cref{sec:full-lang}, having the same expressive power as the full language (this will become apparent in \cref{sec:full-lang}).
\graphalg{} Core not only serves as a vehicle for theoretical analysis:
it is also used as an intermediate representation inside the \graphalg{} compiler~\cite{de_graaf_wildarchgraphalg_2026}.

\graphalg{} Core makes precise the following properties of the language that are irrelevant for theoretical analysis, but are necessary for an executable specification:
\begin{itemize}
    \item The repertoire $\Omega$ is defined by extending the grammar to include syntax for defining pointwise functions.
          An expression in this syntax is called a \emph{scalar expression}.
          While $\Omega$ is not finite, execution platforms only need to support the fixed set of operations included in the syntax to execute arbitrary scalar expressions.

    \item We fix the set of allowed semirings $\texttt{SRing}$ to a number of semirings that are both highly relevant to the implementation of graph algorithms and practical to support across different execution platforms.
\end{itemize}

\subsection{Syntax}
The \graphalg{} Core language syntax is obtained from muse-$\ML$, with the following additional syntax:

\begin{align*}
    f ::= & (C_1,\ldots,C_n) \ e            & \text{(pointwise function)} \\
    e ::= & \ C                             & \text{(scalar variable)}    \\
    |     & \ r(l)                          & \text{(scalar literal)}     \\
    |     & \ e_1 \ \{+,\cdot,-,/,=\} \ e_2 & \text{(scalar arithmetic)}  \\
    |     & \ \cast(r, e)                   & \text{(scalar cast)}        \\
\end{align*}

We use the following notation:
\begin{itemize}
    \item $e$ represents a scalar expression
    \item $C$ names a scalar value defined in a pointwise function.
    \item $l$ is a literal value such as $42$, $3.14$ or `true'.
    \item $r$ is a semiring from the set:
          \[
              \texttt{SRing} = \{ \ringBool, \ringInt, \ringReal, \ringTropInt, \ringTropReal, \ringTropMaxInt, \ringTropMaxReal \}
          \]
          This includes the boolean, integer and real semirings $\ringBool$, $\ringInt$ and $\ringReal$, respectively.
          For the integers and real numbers, their tropical min-plus and max-plus semirings are included as well.
          The min-plus semirings define the addition operator as $\min$ and multiplication as $+$, while the max-plus semirings use $\max$ and $+$.
\end{itemize}

\subsection{Type System}
The typing rules match those of muse-$\ML$ except for $\Apply[f]$, where we now also need to check and infer the type of $f$.

\medskip
\begin{tabular}{lc}
    $\ttypes(\Apply[(C_1,\ldots,C_k)\ e](E_1,\ldots ,E_k)):= \alpha \times \beta \times R_o$, if $\ttypes(E_i) =$   \\
    $\alpha \times \beta \times R_i$ for all $i \in [1 \ldots k]$ and $\ttype_{\Sch_f}(e) = 1 \times 1 \times R_o$, \\
    where $\Sch_f = \{ C_1 := 1 \times 1 \times R_1,\ldots,C_k := 1 \times 1 \times R_k \}$;
\end{tabular}
\medskip

Typing rules for scalar expressions are presented below.
Because scalar expressions always produce a $1 \times 1$ matrix, we write the type as $R$ instead of $1 \times 1 \times R$.

\begin{itemize}
    \item $\ttypes(C):= \ttypes(C)$, for a matrix variable $C\in\Mnam$;
    \item $\ttypes(r(l)):= r$ if $r \in \texttt{SRing}$;
    \item $\ttypes(e_1 + e_2):= R$ if $\ttypes(e_1) = \ttypes(e_2) = R$;
    \item $\ttypes(e_1 \cdot e_2):= R$ if $\ttypes(e_1) = \ttypes(e_2) = R$;
    \item $\ttypes(e_1 - e_2):= R$ if $\ttypes(e_1) = \ttypes(e_2) = R \in \{ \ringInt, \ringReal \}$;
    \item $\ttypes(e_1 / e_2):= \ringReal$ if $\ttypes(e_1) = \ttypes(e_2) = \ringReal$;
    \item $\ttypes(e_1 = e_2):= \ringBool$ if $\ttypes(e_1) = \ttypes(e_2) = R$; and
    \item $\ttypes(\cast(r,e)):= r$ if $r \in \texttt{SRing}$ and $\ttypes(e) = R$.
\end{itemize}

\subsection{Semantics}
Like the typing rules above, \graphalg{} Core largely follows the semantics of muse-$\ML$.
The semantics of $\Apply[f]$ and scalar expressions are defined as follows:

\begin{itemize}
    \item $\sem{\Apply[(C_1,\ldots,C_k)\ e](E_1,\ldots,E_k)}{\I}$ is a matrix $A$ of the same size as $\sem{e_1}{\I}$,
          and where $A_{ij}$ has the value $\sem{e}{\I_{ij}}$, with\\
          $\I_{ij} = \{ C_1 := \sem{E_1}{\I}_{ij},\ldots, C_k := \sem{E_k}{\I}_{ij} \}$;
    \item $\sem{r(l)}{\I} := l$
    \item $\sem{e_1 + e_2}{\I} := \sem{e_1}{\I} \oplus_R \sem{e_2}{\I}$, with $R$ the semiring of $\sem{e_1}{\I}$.
    \item $\sem{e_1 \cdot e_2}{\I} := \sem{e_1}{\I} \otimes_R \sem{e_2}{\I}$, with $R$ the semiring of $\sem{e_1}{\I}$
    \item $\sem{e_1 - e_2}{\I} := \sem{e_1}{\I} - \sem{e_2}{\I}$
    \item $\sem{e_1 / e_2}{\I} := \sem{e_1}{\I} / \sem{e_2}{\I}$
    \item $\sem{e_1 = e_2}{\I} := \sem{e_1}{\I} = \sem{e_2}{\I}$
    \item $\sem{\cast(r, e)}{\I} := \cast(R, r, \sem{e}{\I})$, whose definition is given below.
\end{itemize}

\subsubsection{Casting}
Function $\cast(R, T, v)$ casts a value from semiring $R$ to semiring $T$.
A key rule in casting is \emph{additive identity preservation}: If $v = 0_R$, the value after casting is $0_T$ (see the second rule below).
This rule takes precedence over all rules that follow it.
For the remaining casting rules, only the domain of the semiring is significant.
To simplify the rule set, we write for example $\mathbb{Z}$ where any semiring over integers is accepted.

\begin{align*}
    \cast(R, R, v)                   & = v                        \\
    \cast(R, r, 0_R)                 & = 0_r                      \\
    \cast(R, \mathbb{B}, v)          & = 1 \text{ if } v \neq 0_R \\
    \cast(\mathbb{B}, r, 1)          & = 1_r                      \\
    \cast(\mathbb{Z}, \mathbb{R}, v) & = v                        \\
    \cast(\mathbb{R}, \mathbb{Z}, v) & = \lfloor v \rfloor        \\
\end{align*}

Preserving additive identities across casts is an important property for program optimization of matrices with sparse representations.
\graphalg{} is designed to support efficient execution over sparse matrix representations that store only non-zero elements (values that are not the additive identity).
Because of additive identity preservation, casting a sparse matrix to a different semiring only requires updating non-zero elements, because implicit zeros in the input remain zero in the output.

\section{The Full GraphAlg Language}\label{sec:full-lang}
We can now give a definition of the `full' \graphalg{} language that users interact with.
The \graphalg{} language can be seen as syntactic sugar over \graphalg{} Core, with the aim of presenting an interface that is familiar to many programmers and facilitates writing algorithms in a concise way.
Key differences with \graphalg{} Core include:
\begin{itemize}
    \item An imperative programming style rather than the strictly functional nature of \graphalg{} Core;
    \item \graphalg{} programs are composed of functions (where \graphalg{} Core only has expressions);
    \item Dedicated syntax for commonly used patterns, such as aggregations across rows and/or columns.
          These are inspired by the GraphBLAS~\cite{davis_algorithm_2019}, an effort to standardize building blocks (but not a language) for writing graph algorithms in linear algebra.
    \item Alternative syntax for operations that fall outside the ASCII character set, such as $e^T$.
\end{itemize}

In this section we define \graphalg{} in terms of \graphalg{} Core.
In this way, we obtain that \graphalg{} has the same expressive power as \graphalg{} Core.
We derive two key properties of \graphalg{} that enable the transformation of arbitrary \graphalg{} programs into \graphalg{} Core.
Firstly, we show that the imperative programming style with statements and mutable variables that is idiomatic in \graphalg{} translates into purely functional \graphalg{} Core expressions (\cref{sec:stmts}).
Secondly, we establish that any non-scalar \graphalg{} expression with scalar inputs can be written as a \emph{scalar expression} (\cref{thm:scalarize}), and use this to rewrite references between functions (\cref{sec:func-refs}).

The \graphalg{} language is based on \graphalg{} Core: the expressions $\diag(e)$ and $\pickAny(e)$ have the same semantics in \graphalg{} as they do in \graphalg{} Core.
We also carry over the transpose and matrix multiplication operations, but denote them as $e \texttt{.T}$ and $e_1 \texttt{ * } e_2$, respectively, to stay within the ASCII character set.
The expressions $\one(e)$, $\Apply[f](e)$ and $\texttt{for}\ \{\ldots\}(\ldots)$ are not included in \graphalg{}, as they are subsumed by the syntax introduced in the sections below.

\subsection{Summation}
Since summation along one or more dimensions is commonly used, \graphalg{} has dedicated syntax for this:

\begin{align*}
    e & ::= \texttt{reduceRows(} e \texttt{)}       &  & \text{(sum per row)}    \\
      & \mid\quad \texttt{reduceCols(} e \texttt{)} &  & \text{(sum per column)} \\
      & \mid\quad \texttt{reduce(} e \texttt{)}     &  & \text{(sum to scalar)}  \\
\end{align*}

Naming follows that of the equivalent operations in the GraphBLAS.
These operations can be defined in terms of multiplication with one-vectors:

\begin{align*}
    \texttt{reduceRows(} e \texttt{)} & \Rightarrow e \cdot \one(\transp{e})                    \\
    \texttt{reduceCols(} e \texttt{)} & \Rightarrow \transp{\one(e)} \cdot e                    \\
    \texttt{reduce(} e \texttt{)}     & \Rightarrow \texttt{reduceRows}(\texttt{reduceCols}(e)) \\
\end{align*}

\subsection{Scalar Operators}
\graphalg{} supports all scalar expressions of graphalg{} Core, with the following additions and substitutions:

\begin{align*}
    e & ::= \texttt{!} e                 &  & \text{(logical negation)}        \\
      & \mid \quad \texttt{-} e          &  & \text{(numerical negation)}      \\
      & \mid \quad e_1 \texttt{ * } e_2  &  & \text{(scalar multiply)}         \\
      & \mid \quad e_1 \texttt{ == } e_2 &  & \text{(scalar equality)}         \\
      & \mid \quad e_1 \texttt{ != } e_2 &  & \text{(scalar inequality)}       \\
      & \mid \quad \texttt{zero}(r)      &  & \text{(additive identity)}       \\
      & \mid \quad \texttt{one}(r)       &  & \text{(multiplicative identity)} \\
\end{align*}
where $r$ is a semiring.
We include an alternative syntax for scalar multiplication (to fit the ASCII character set) and scalar equality (to align with commonly-used programming languages), the other operators do not have a direct equivalent in \graphalg{} Core.
The above operators are defined in terms of \graphalg{} Core as follows:

\begin{align*}
    \texttt{!} e          & \Rightarrow e = \ringBool(\text{false}) \\
    \texttt{-} e          & \Rightarrow (e - e) - e                 \\
    e_1 \texttt{ * } e_2  & \Rightarrow e_1 \cdot e_2               \\
    e_1 \texttt{ == } e_2 & \Rightarrow e_1 = e_2                   \\
    e_1 \texttt{ != } e_2 & \Rightarrow \texttt{!}(e_1 = e_2)       \\
    \texttt{zero}(r)      & \Rightarrow r(0_r)                      \\
    \texttt{one}(r)       & \Rightarrow r(1_r)                      \\
\end{align*}
where $0_r$ and $1_r$ represent the additive and multiplicative identities of semiring $r$, respectively.
We add the operations \texttt{zero} and \texttt{one} because some semirings have identity elements that cannot be expressed in ASCII (such as $\infty$).
The expression $\texttt{-} e$ can alternatively be written as $r(0_r) - e$, where $r$ is the semiring of $e$.

\subsection{Pointwise Operators}
The following syntax allows pointwise application of scalar operators to full matrices:

\begin{align*}
    e & ::=  e_1 \texttt{ (.} \{ \texttt{+}, \texttt{-}, \texttt{*}, \texttt{/}, \texttt{==}, \texttt{!=} \} \texttt{) } e_2 &  & \text{(pointwise operator)} \\
      & \mid\quad \texttt{cast<} r \texttt{>}(e)                                                                             &  & \text{(pointwise cast)}     \\
\end{align*}

These constructs lower to pointwise function application:

\begin{align*}
    e_1 \texttt{(.} \diamond \texttt{) } e_2 & \Rightarrow \Apply[(v_1, v_2)\ v_1 \diamond v_2](e_1, e_2) \\
    \texttt{cast<} r \texttt{>}(e)           & \Rightarrow \Apply[(v)\ \cast(r, v)](e)
\end{align*}

\subsection{Matrix Properties}
\graphalg{} provides native support for common matrix properties:

\begin{align*}
    e & ::=  e \texttt{.nrows}       &  & \text{(number of rows)}        \\
      & \mid \quad e \texttt{.ncols} &  & \text{(number of columns)}     \\
      & \mid \quad e \texttt{.nvals} &  & \text{(count non-zero values)} \\
\end{align*}

The names of the properties follow the GraphBLAS.
Matrix property expressions can be rewritten as:

\begin{align*}
    e \texttt{.nrows} & \Rightarrow \texttt{reduce}(\one(\texttt{cast<int>}(e)))                \\
    e \texttt{.ncols} & \Rightarrow (\transp{e}) \texttt{.nrows}                                \\
    e \texttt{.nvals} & \Rightarrow \texttt{reduce}(\texttt{cast<int>}(\texttt{cast<bool>}(e))) \\
\end{align*}

\subsection{Creating Matrices and Vectors}
We introduce operations to create matrices of different sizes:

\begin{align*}
    e & ::= \texttt{Matrix<} r \texttt{>} (e_1, e_2)   &  & \text{(zero matrix)} \\
      & \mid \quad \texttt{Vector<} r \texttt{>} (e_1) &  & \text{(zero vector)} \\
\end{align*}

In a valid \graphalg{} program, the input expressions $e_1$ and $e_2$ must refer to the dimensions of an existing matrix, that is, they must be of the form $e \texttt{.nrows}$ or $e \texttt{.ncols}$.
To work with this type of expressions, we define a helper function $\mathsf{ndim}$ as:

\begin{align*}
    \mathsf{ndim}(e \texttt{.nrows}) & = \one(e)          \\
    \mathsf{ndim}(e \texttt{.ncols}) & = \one(\transp{e}) \\
\end{align*}

We now give the rewrite rules for the zero matrix and vector:

\begin{align*}
    \texttt{Matrix<} r \texttt{>} (e_1, e_2) & \Rightarrow \Apply[(a)\ \texttt{zero}(r)](\mathsf{ndim}(e_1) \cdot \transp{\mathsf{ndim}(e_2)}) \\
    \texttt{Vector<} r \texttt{>} (e)        & \Rightarrow \Apply[(a)\ \texttt{zero}(r)](\mathsf{ndim}(e))                                     \\
\end{align*}

\subsection{Functions}\label{sec:full-funcs}
Where in \graphalg{} Core the type of an expression $e$ relies on a given schema $\Sch$, in \graphalg{} expressions are contained in \emph{functions} that provide a syntax for defining such schemas.
A \graphalg{} \emph{program} is a collection of such functions.
See \cref{fig:example-program} for an example program.

\begin{figure}[h]
    \lstset{basicstyle=\ttfamily, language=GraphAlg}
\begin{lstlisting}
func NotEquals(v:int, n:int) -> bool { return v != n; }

func SumUnless(v:Vector<s,int>, n:int) -> int {
  vfilt = select(NotEquals, v, n);
  return reduce(vfilt);
}
\end{lstlisting}
    \caption{
        An example \graphalg{} program with multiple functions.
        Function \texttt{SumUnless} sums the values in \texttt{v} that do not equal \texttt{n}.
        Function \texttt{NotEquals} is a helper function used in \texttt{SumUnless}.
    }
    \label{fig:example-program}
\end{figure}

The syntax for programs follows the grammar:

\begin{align*}
    P & ::= F_1 \cdots F_n                                                                                                                     &  & \text{(program)}             \\
    F & ::= \texttt{func } f(M_1: t_1, \ldots, M_n: t_n) \texttt{ -> } t \texttt{ \{}                                                                                            \\
      & \quad \quad \quad S_1                                                                                                                                                    \\
      & \quad \quad \quad \, \vdots                                                                                                                                              \\
      & \quad \quad \quad S_k                                                                                                                                                    \\
      & \quad \quad \quad \texttt{return } e \texttt{;}                                                                                                                          \\
      & \quad \quad \texttt{\}}                                                                                                                &  & \text{(function)}            \\
    t & ::= r                                                                                                                                  &  & \text{(scalar)}              \\
      & \mid\quad \texttt{Matrix<} s_1 \texttt{, } s_2 \texttt{, } r \texttt{>}                                                                &  & \text{(matrix)}              \\
      & \mid\quad \texttt{Vector<} s \texttt{, } r \texttt{>}                                                                                  &  & \text{(vector)}              \\
    s & :: = 1                                                                                                                                 &  & \text{(explicit `1' size)}   \\
      & \mid\quad \alpha                                                                                                                       &  & \text{(size symbol)}         \\
    S & ::= M\ (\texttt{<} (\texttt{!})^? M_m \texttt{>})^?\ (\texttt{[:} (\texttt{,:})^?\texttt{]})^? \texttt{ = } e \texttt{;}               &  & \text{(assignment)}          \\
      & \mid\quad M\ \texttt{+= } e \texttt{;}                                                                                                 &  & \text{(accumulate)}          \\
      & \mid\quad \texttt{for } M \texttt{ in } e_1 \texttt{:} e_2 \texttt{ \{} S_1 \cdots S_k \texttt{\} } (\texttt{until } e_3 \texttt{;})^? &  & \text{(constant range loop)} \\
      & \mid\quad \texttt{for } M \texttt{ in } e_1 \texttt{ \{} S_1 \cdots S_k \texttt{\} } (\texttt{until } e_2 \texttt{;})^?                &  & \text{(matrix size loop)}    \\
\end{align*}

For now, we only consider functions of the form:
\[
    \texttt{func } f(M_1: t_1, \ldots, M_n: t_n) \texttt{ -> } t \texttt{ \{ return } e \texttt{; \}}
\]

We defer the definition of the semantics of statement syntax $S$ to \cref{sec:stmts}, where we show that every function in \graphalg{} can be written in the above form (that is, without the $S_1 \cdots S_k$).
The parameters of a function define the schema for expression $e$ as follows.
For a parameter $M:t$, its type $\ttype(t)$ is defined as:

\begin{align*}
    \ttype(r)                                                             & := 1 \times 1 \times r     \\
    \ttype(\texttt{Matrix<} s_1 \texttt{, } s_2 \texttt{, } r \texttt{>}) & := s_1 \times s_2 \times r \\
    \ttype(\texttt{Vector<} s \texttt{, } r \texttt{>})                   & := s \times 1 \times r     \\
\end{align*}

Syntactically, size symbols are an alphanumeric character sequence that start with a letter (to forbid size symbols such as `1' or `42').
Semiring $r$ can be written using the following ASCII-compatible syntax.
\texttt{bool}, \texttt{int} and \texttt{real} denote $\ringBool$, $\ringInt$ and $\ringReal$, respectively, while
tropical semirings are written as \texttt{trop\_int}, \texttt{trop\_real}, \texttt{trop\_max\_int} and \texttt{trop\_max\_real} for $\ringTropInt$, $\ringTropReal$, $\ringTropMaxInt$ and $\ringTropMaxReal$.

The schema for an expression $e$ in a function $f$ assuming a list of parameters $M_1 : t_1, \ldots, M_n : t_n$ is defined as $\Sch_f(M_i) := \ttype(t_i)$, with $i \in \{ 1,\ldots,n \}$.
Function $f$ is \emph{well-typed} if the type of expression $e$ matches the declared return type of the function, that is, $\ttype_{\Sch_f}(e) = \ttype(t)$.
The full type (rather than the return type) of a function consists of the types of its parameters and its return type:

\begin{multline*}
    \ttype(\texttt{func } f(M_1: t_1, \ldots, M_n: t_n) \texttt{ -> } t \texttt{ \{}  \cdots \texttt{\}}) := \\
    \ttype(t_1) \times \cdots \times \ttype(t_n) \times \ttype(t)
\end{multline*}

The semantics of a \graphalg{} function $f$ are expressed in terms of \graphalg{} Core.
Let $A_1,\ldots,A_n$ be the input matrices to bind to the parameters of $f$, and let $\texttt{return } e \texttt{;}$ be the body of $f$.
Then, let $e'$ be the \graphalg{} Core expression obtained by applying the rewrite rules defined in the previous sections to $e$.
We define the instance $\I(M_i) := A_i$, where $i \in \{ 1,\ldots,n \}$.
Finally, the value of function $f$ applied to $A_1,\ldots,A_n$ is $\sem{e'}{\I}$, having type $\ttype(t)$.
\cref{cor:full-to-core} follows directly from the above definition.

\begin{corollary}\label{cor:full-to-core}
    Let $f$ be a \graphalg{} function with $\ttype(f) := \tau_1 \times \cdots \times \tau_n \times \tau$,
    and let $A_1,\ldots,A_n$ be matrices such that $A_i$ has type $\tau_i$ for $i \in \{ 1,\ldots,n \}$.
    Then there exists a \graphalg{} Core expression $e$ that evaluates $f(A_1,\ldots,A_n)$.
\end{corollary}

\subsection{Function References}\label{sec:func-refs}
\graphalg{} introduces syntax to apply a function defined earlier in the same program, rather than providing the function definition inline as in \graphalg{} Core.

\begin{align*}
    e & ::= e_1 \texttt{ (.} f \texttt{) } e_2                                      &  & \text{(pointwise binary function)} \\
      & \mid\quad \texttt{apply(} f \texttt{, } e_1 (\texttt{, } e_2)^? \texttt{)}  &  & \text{(pointwise function)}        \\
      & \mid\quad \texttt{select(} f \texttt{, } e_1 (\texttt{, } e_2)^? \texttt{)} &  & \text{(pointwise selection)}       \\
\end{align*}

This notation is modeled after similar operations in the GraphBLAS.
Intuitively, the pointwise operations above apply function $f$ to every element in the matrix that $e_1$ evaluates to.
In the example program given in \cref{fig:example-program}, \texttt{select} behaves as follows.
For every position $\mathtt{v}_{ij}$, we evaluate $\texttt{NotEquals}(\mathtt{v}_{ij}, \mathtt{n})$.
If the result is true, set $\mathtt{vfilt}_{ij} = \mathtt{v}_{ij}$, or zero otherwise.

Let $\mathcal{F}$ be the function that maps the names of functions to their types, such that for any function $f$ defined before expression $e$ we have:

\[
    \mathcal{F}(f) := \ttype(\texttt{func } f(M_1: t_1, \ldots, M_n: t_n) \texttt{ -> } t \texttt{ \{}  \cdots \texttt{\}})
\]

Then the typing rules for the above constructs are:

\begin{itemize}
    \item $\ttypes(e_1 \texttt{ (.} f \texttt{) } e_2) := \alpha \times \beta \times R_3$ if $\ttypes(e_1) = \alpha \times \beta \times R_1$ and $\ttypes(e_2) = \alpha \times \beta \times R_2$ and $\mathcal{F}(f) = (1 \times 1 \times R_1) \times (1 \times 1 \times R_2) \times (1 \times 1 \times R_3)$;
    \item $\ttypes(\texttt{apply(} f \texttt{, } e \texttt{)}) := \alpha \times \beta \times R_2$ if $\ttypes(e) = \alpha \times \beta \times R_1$ and $\mathcal{F}(f) = (1 \times 1 \times R_1) \times (1 \times 1 \times R_2)$;
    \item $\ttypes(\texttt{apply(} f \texttt{, } e_1 \texttt{, } e_2 \texttt{)}) := \alpha \times \beta \times R_3$ if $\ttypes(e_1) = \alpha \times \beta \times R_1$ and $\ttypes(e_2) = 1 \times 1 \times R_2$ and $\mathcal{F}(f) = (1 \times 1 \times R_1) \times (1 \times 1 \times R_2) \times (1 \times 1 \times R_3)$;
    \item $\ttypes(\texttt{select(} f \texttt{, } e \texttt{)}) := \alpha \times \beta \times R$ if $\ttypes(e) = \alpha \times \beta \times R$ and $\mathcal{F}(f) = (1 \times 1 \times R) \times (1 \times 1 \times \ringBool)$;
    \item $\ttypes(\texttt{select(} f \texttt{, } e_1 \texttt{, } e_2 \texttt{)}) := \alpha \times \beta \times R_1$ if $\ttypes(e_1) = \alpha \times \beta \times R_1$ and $\ttypes(e_2) = 1 \times 1 \times R_2$ and $\mathcal{F}(f) = (1 \times 1 \times R_1) \times (1 \times 1 \times R_2) \times (1 \times 1 \times \ringBool)$;
\end{itemize}

Note that all input and output dimensions of referenced functions must be the explicit `1'.

Next, we show how under this assumption any \graphalg{} function can be rewritten into a pointwise function in $\Omega$, a property we will later use to rewrite the syntax presented above.

\begin{theorem}\label{thm:scalarize}
    Let $e$ be a \graphalg{} expression that is well-typed under schema $\Sch$.
    If for all $M \in \Mvar$ we have $\Sch(M) = 1 \times 1 \times R$, then there exists a scalar expression that is equivalent to $e$.
\end{theorem}
\begin{proof}
    Recall that $\ML$, and by extension \graphalg{}, can only construct matrices based on dimensions present in the input schema.
    It follows that any sub-expression inside $e$ operates exclusively on 1-by-1 (scalar) matrices.
    By induction on the structure of expression $e$, we show that an equivalent scalar expression can be constructed if all inputs are 1-by-1 matrices.
    We only consider expressions from \graphalg{} Core, as the additional operations added in \graphalg{} can be expressed in terms of \graphalg{} Core.

    Matrix transpose, $\diag$ and $\pickAny$ over a 1-by-1 matrix return the original matrix.
    Matrix multiplication simplifies to scalar multiplication:

    \begin{align*}
        \transp{e}    & \Rightarrow e                    \\
        \diag(e)      & \Rightarrow e                    \\
        \pickAny(e)   & \Rightarrow e                    \\
        e_1 \cdot e_2 & \Rightarrow e_1 \texttt{ * } e_2 \\
    \end{align*}

    Function application is simplified by inlining the pointwise function:

    \[
        \Apply[(C_1,\ldots,C_k)\ e](e_1,\ldots,e_k) \Rightarrow e[C_1/e_1,\ldots,C_k/e_k]
    \]
    where by $e[C_1/e_1,\ldots,C_k/e_k]$ we denote the expression obtained by substitution of $e_i$ for any occurrence of $C_i$ in $e$.
    Lastly, \texttt{for} loops run for exactly one iteration, and therefore simplify to:

    \[
        \texttt{for}\, \{ M_1 := e_1', \ldots, M_m := e_m' \}(e_d, e_1, \ldots, e_m) \Rightarrow e_1'[M_1/e_1,\ldots,M_m/e_m]
    \]

\end{proof}

We define a helper function to fill a matrix with the same value at every position:

\[
    \mathsf{fill}(e_1, e_2) = \one(\Apply[(C)\ 1_R](e_1)) \cdot e_2 \cdot \Apply[(C)\ 1_R](\transp{\one(\transp{e_1})})
\]
where $1_R$ is the multiplicative identity of the semiring over which $e_2$ is defined.
The expression $\mathsf{fill}(e_1, e_2)$ evaluates to a matrix with the dimensions of $e_1$ and the value of $e_2$ at all positions.
Note that $e_2$ must be a $1 \times 1$ matrix for $\mathsf{fill}(e_1, e_2)$ to be well-typed.

Let $\texttt{func } f(M_1: t_1, \ldots, M_n: t_n) \texttt{ -> } t \texttt{ \{}  \cdots \texttt{\}}$ be an arbitrary function with 1-by-1 input and output types.
By \cref{thm:scalarize}, there is an effectively computable corresponding scalar function $f' := (C_1,\ldots,C_n)\ e'$ that is equivalent to $f$.
We can now give the following rewrite rules leveraging this correspondence:

\begin{align*}
    e_1 \texttt{ (.} f \texttt{) } e_2                            & \Rightarrow \Apply[f'](e_1, e_2)                                                        \\
    \texttt{apply(} f \texttt{, } e \texttt{)}                    & \Rightarrow \Apply[f'](e)                                                               \\
    \texttt{apply(} f \texttt{, } e_1 \texttt{, } e_2 \texttt{)}  & \Rightarrow \Apply[f'](e_1, \mathsf{fill}(e_1, e_2))                                    \\
    \texttt{select(} f \texttt{, } e \texttt{)}                   & \Rightarrow \Apply[(C_1)\ C_1 \cdot \cast(R, e')](e)                                    \\
    \texttt{select(} f \texttt{, } e_1 \texttt{, } e_2 \texttt{)} & \Rightarrow  \Apply[(C_1, C_2)\ C_1 \cdot \cast(R_1, e')](e_1, \mathsf{fill}(e_1, e_2)) \\
\end{align*}
where $R, R_1, R_2$ are the semirings of $e,e_1,e_2$, respectively.

Function reference syntax subsumes $\Apply[f]$ from \graphalg{} Core.
An important difference is that function application in \graphalg{} is limited to function with arity 2, whereas function $f$ in $\Apply[f]$ can have arbitrary arity.
However, this does not limit the expressive power of the language, because scalar operators have an arity of at most 2.
As an example, consider the \graphalg{} Core expression $\Apply[(a,b,c)\ a \cdot b + c](e_1, e_2, e_3)$.
Clearly, the \graphalg{} expression $e_1 \text{ (.*) } e_2 \texttt{ (.+) } e_3$ is equivalent.

\subsection{Statements}\label{sec:stmts}
We now turn to the syntax for statements inside function bodies, as introduced in \cref{sec:full-funcs}.
We show that a function body of the form $S\ \texttt{return } e;$ can be written as a single \graphalg{} expression.
By successively applying such rewrites, function bodies with an arbitrary number of statements can be reduced to as single expression, e.g., $S_1\ S_2\ \texttt{return } e \Rightarrow S_1\ \texttt{return } e' \Rightarrow \texttt{return } e''$.

We start with the simplest case, which is plain assignment:

\[
    M \texttt{ = } e_1 \texttt{; return } e_2 \texttt{;} \Rightarrow \texttt{return } \LetIn{M}{e_1}{e_2} \texttt{;}
\]

All other statements can be expressed as plain assignments.
Accumulation becomes pointwise addition with the previous definition of $M$:

\[
    M \texttt{ += } e \texttt{;} \Rightarrow M \texttt{ = } M \texttt{ (.+) } e \texttt{;}
\]

Assignment with fill syntax adjusts the dimensions of $e$ to match those of $M$:

\begin{align*}
    M\texttt{[:]} \texttt{ = } e \texttt{;}    & \Rightarrow M \texttt{ = } \one(M) \cdot e \texttt{;}     \\
    M\texttt{[:, :]} \texttt{ = } e \texttt{;} & \Rightarrow M \texttt{ = } \mathsf{fill}(M, e) \texttt{;} \\
\end{align*}

Fill syntax subsumes the $\one(e)$ expression of \graphalg{} Core.
Given an expression $e$ with semiring $R$, $\one(e)$ can simulated as:

\[
    M = \texttt{Vector<}R\texttt{>}(e.\texttt{nrows})\texttt{; }
    M\texttt{[:]} = \texttt{one}(R)\texttt{;}
\]
where $M$ is equivalent to $\one(e)$.

Masking performs pointwise conditional assignment.
To rewrite masking, we use the following helper function:

\[
    \mathsf{mask}(B, M, C) = \texttt{cast<}R \texttt{>}(M) \cdot C + \texttt{cast<}R \texttt{>}(\texttt{!}M) \cdot B
\]
where $R$ is the semiring of $B$ (and of $C$).
We choose this definition to obtain the following properties:

\begin{align*}
    \mathsf{mask}(B, \texttt{true}\ , C) = C \\
    \mathsf{mask}(B, \texttt{false}, C) = B
\end{align*}

The rewrite rules are then defined as:

\begin{align*}
    M_1\texttt{<}M_2\texttt{>} \texttt{ = } e \texttt{;}  & \Rightarrow M_1 \texttt{ = } \Apply[\mathsf{mask}](M_1, \cast(\ringBool, M_2), e) \texttt{;}           \\
    M_1\texttt{<!}M_2\texttt{>} \texttt{ = } e \texttt{;} & \Rightarrow M_1 \texttt{ = } \Apply[\mathsf{mask}](M_1, \texttt{!}\cast(\ringBool, M_2), e) \texttt{;} \\
\end{align*}

Fill and mask syntax may be combined in a single statement.
This is equivalent to filling followed by masking.
With $T$ a fresh matrix variable, the rewrite rule to split them into separate statements is:

\begin{tabbing}
    LLL0\=LLL1\=L2\=L3\=L4 \kill
    \>$M_1\texttt{<}M_2\texttt{>[:]} \texttt{ = }  e \texttt{;} \Rightarrow$ \\
    \>\>$T \texttt{ = } M_1 \texttt{;}$ \\
    \>\>$T \texttt{[:]} \texttt{ = } e \texttt{;}$ \\
    \>\>$M_1\texttt{<}M_2\texttt{>} \texttt{ = } T \texttt{;}$ \\
\end{tabbing}

The semantics of fill and mask syntax are modeled after the GraphBLAS.

\graphalg{} supports two types of loop statements: loops over a constant range, and loops over matrix dimensions.
We treat these two loops constructs separately, as their rewrites are mostly distinct.
We first give rewrites for the simpler cases without the optional \texttt{until} syntax, and afterwards extend the rewrites to handle this case too.

Consider constant range loops of the form $\texttt{for } M \texttt{ in } e_1 \texttt{:} e_2 \texttt{ \{} S_1 \cdots S_k \texttt{\}}$.
The expressions $e_1$ and $e_2$ must be integer literals (or expressions that simplify to literals through constant propagation) that represent a non-empty range.
This requirement allows us to represent loops in \graphalg{} Core by \emph{unrolling} them.
Let $e_1:e_2$ represent range $[l, h)$,
then the loop can be represented as the following sequence of statements:

\begin{table*}[h]
    \begin{tabular}{l l}
        \qquad $\texttt{for } M \texttt{ in } e_1 \texttt{:} e_2 \texttt{ \{} S_1 \cdots S_k \texttt{\}} \Rightarrow$ \\
        \qquad $M \texttt{ = } e_1 \texttt{;}$                                                                        \\
        \qquad $S_1 \cdots S_k$ & \multirow{2}{*}{\Big \} $h-l$ times}                                                \\
        \qquad $M \texttt{ = } \Apply[(C)\ C + \ringInt(1)](M) \texttt{;}$                                            \\
    \end{tabular}
\end{table*}

\noindent where $M \texttt{ = } \Apply[\cdots](M)$ increments iteration variable $M$ after every loop iteration.

Loops of matrix dimensions can be represented using the loop expression from \graphalg{} Core.
Without loss of generality, let loop body $S_1,\ldots,S_k$ be a sequence of plain assignments (as established at the start of this section, all statements can be rewritten into sequences of plain assignments).
Let $M_1,\ldots,M_n$ be the matrix variables assigned to in $S_1,\ldots,S_k$ with a previous definition before the loop.
For each $i \in \{ 1, \ldots, n \}$, let $e_i$ be the \graphalg{} Core expression obtained by rewriting the function body $S_1,\ldots,S_k \texttt{ return } M_i \texttt{;}$.
Finally, the loop can then be rewritten as:

\begin{tabbing}
    LLL0\=LLL1\=L2\=L3\=L4 \kill
    \>$\texttt{for } M \texttt{ in } e \texttt{ \{} S_1 \cdots S_k \texttt{\}} \Rightarrow$ \\
    \>$T_1 \texttt{ = } M_1 \texttt{;} \cdots T_n \texttt{ = } M_n \texttt{;}$ \\
    \>$M_1 \texttt{ = } \texttt{for}\, \{ M_1 := e_1, \ldots, M_n := e_n \}(\mathsf{ndim}(e), T_1, \ldots, T_n)$ \\
    \>$M_2 \texttt{ = } \texttt{for}\, \{ $ \\
    \>\>$M_2 := e_2, M_1 := e_1, M_3 := e_3, \ldots, M_n := e_n$ \\
    \>$\}(\mathsf{ndim}(e), T_2, T_1, T_3, \ldots, T_n)$ \\
    \>$\vdots$ \\
    \>$M_n \texttt{ = } \texttt{for}\, \{$\\
    \>\>$M_n := e_n, M_1 := e_1, \ldots, M_{n-1} := e_{n-1}$\\
    \>$\}(\mathsf{ndim}(e), T_n, T_1, \ldots, T_{n-1})$ \\
\end{tabbing}
where $T_1,\ldots T_n$ are fresh variables.
The loop is duplicated for every $M_i$, as a loop expression can only have a single result.
Every instance of the loop uses a different first loop binding to reflect the final desired loop result.
Because the ordering of bindings does not otherwise influence the behaviour of the loop, and all operations in \graphalg{} Core are free of side effects and deterministic, loops can be safely duplicated.

\graphalg{} supports early loop termination using the \texttt{until} construct.
While a real \graphalg{} implementation should typically treat such a break condition specially for performance reasons, it can also be simulated.
Assume an arbitrary loop statement with a break condition $\texttt{for } M \texttt{ in } \cdots \texttt{ \{} \cdots \texttt{\} } \texttt{until } e \texttt{;}$.
Without loss of generality, let the loop body be a sequence of plain assignments $M_1 \texttt{ = } e_1 \texttt{;} \cdots M_k \texttt{ = } e_k \texttt{;}$.
This loop can be rewritten without \texttt{until} as:

\begin{tabbing}
    LLL0\=LLL1\=L2\=L3\=L4 \kill
    \>$\texttt{for } M \texttt{ in } \cdots \texttt{ \{} M_1 \texttt{ = } e_1 \texttt{;} \cdots M_k \texttt{ = } e_k \texttt{;\} } \texttt{until } e \texttt{;} \Rightarrow$ \\
    \>$U \texttt{ = } \ringBool(\text{false}) \texttt{;}$ \\
    \>$\texttt{for } M \texttt{ in } \cdots \texttt{ \{}$ \\
    \>\>$U_1 \texttt{ = } \mathsf{fill}(M_1, U) \texttt{;}$ \\
    \>\>$M_1 \texttt{<!} U_1 \texttt{> = } e_1 \texttt{;}$ \\
    \>\>$\ \vdots$ \\
    \>\>$U_k \texttt{ = } \mathsf{fill}(M_k, U) \texttt{;}$ \\
    \>\>$M_k \texttt{<!} U_k \texttt{> = } e_k \texttt{;}$ \\
    \>\>$U \texttt{ = } e \texttt{;}$ \\
    \>$\texttt{\}}$ \\
\end{tabbing}
where $U_1,\ldots,U_k$ and $U$ are fresh variables.
Variable $U$ tracks whether the break condition has been hit.
Each $U_i$ is a mask matrix of the same size as $M_1$ filled with $U$.
By transforming assignment $M_i \texttt{ = } e_i$ into masked assignment $M_i \texttt{<!} U_i \texttt{> = } e_i \texttt{;}$
variable $M_i$ is modified only if the break condition has not yet been hit, preserving the current value otherwise.

Loop syntax in \graphalg{} subsumes the loop expressions of \graphalg{} Core.
Let $\texttt{for}\, \{ M_1 := e_1', \ldots, M_m := e_m' \}(e_d, e_1, \ldots, e_m)$ be an arbitrary loop expression.
We can write this in \graphalg{} as:

\begin{tabbing}
    LLL0\=LLL1\=L2\=L3\=L4 \kill
    \>$M_1 \texttt{ = } e_1 \texttt{;} \cdots M_m \texttt{ = } e_m \texttt{;}$ \\
    \>$\texttt{for } M \texttt{ in } e_d \texttt{.nrows \{}$ \\
    \>\>$T_1 \texttt{ = } e_1' \texttt{;} \cdots T_m \texttt{ = } e_m' \texttt{;}$ \\
    \>\>$M_1 \texttt{ = } T_1 \texttt{;} \cdots M_m \texttt{ = } T_m \texttt{;}$ \\
    \>$\texttt{\}}$ \\
\end{tabbing}
where $T_1,\ldots T_n$ are fresh variables.
After the loop statement, $M_1$ is equivalent to the original \graphalg{} Core expression.

\subsection{Relation to $\ML$}
We are now ready to show that for-$\ML$ extended with support for simultaneous induction (that is, sifor-$\ML$) is powerful enough to simulate \graphalg{} programs.

\begin{corollary}\label{cor:sifor-eq}
    For any \graphalg{} well-typed function $f$ and matrices $A_1,\ldots,A_n$ matching the parameter types of $f$, there exists an expression $e$ in the sifor-$\ML$ language that simulates $f(A_1,\ldots,A_n)$.
\end{corollary}

\begin{figure}[htbp]
    \centering
    \begin{tikzpicture}[
            font=\sffamily\small,
            every node/.style={font=\sffamily\small},
            ann/.style={font=\sffamily\scriptsize\itshape, inner sep=1pt, text=black!60},
            lbox/.style={draw=cbBlue, line width=0.6pt, rounded corners=4pt},
        ]
        \node[lbox, fill=cbBlue!35, minimum width=40mm, minimum height=11mm,
            inner sep=0pt] (ml) {$\ML$};
        \node[ann, above=2mm of ml] (forann) {(for loops)};
        \node[above=0.5mm of forann, font=\sffamily\small] (fortitle) {for-$\ML$};
        \node[fit=(ml)(fortitle)(forann), inner sep=5pt, draw=none] (formlshape) {};
        \node[ann, above=3mm of formlshape.north] (outerann) {(simultaneous induction)};
        \node[above=0.5mm of outerann, font=\sffamily\small] (outertitle2) {$\equiv$ \graphalg{} Core $\equiv$ \graphalg{}};
        \node[above=0.5mm of outertitle2, font=\sffamily\small] (outertitle1) {sifor-$\ML\ \equiv\ $dec-$\ML\ \equiv\ $muse-$\ML$};
        \begin{pgfonlayer}{bgOuter}
            \node[lbox, fill=cbBlue!6, fit=(formlshape)(outertitle1)(outertitle2)(outerann),
                inner sep=7pt] (outer) {};
        \end{pgfonlayer}
        \begin{pgfonlayer}{bgForml}
            \node[lbox, fill=cbBlue!18, fit=(ml)(fortitle)(forann), inner sep=5pt] (forml) {};
        \end{pgfonlayer}
    \end{tikzpicture}
    \caption{
        Expressive-power hierarchy underlying \cref{cor:sifor-eq}.
        $\ML$ is extended by for-$\ML$ (adding a canonical for-loop, \cref{sec:prelim}),
        which is in turn subsumed by sifor-$\ML$ (adding simultaneous induction, \cref{sec:sifor-ml}).
        By \cref{thm:sifor-eq-dec} and \cref{thm:muse-eq-dec}, dec-$\ML$ and muse-$\ML$ have the same expressive power as sifor-$\ML$.
        \graphalg{} Core (\cref{sec:core-lang}) instantiates muse-$\ML$ with a fixed repertoire $\Omega$ and semiring set;
        the containment chain shown here holds when this instantiation propagates to the other languages as described in \cref{thm:sifor-eq-dec,thm:muse-eq-dec}.
        \graphalg{} adds syntactic sugar on top of \graphalg{} Core.
    }
    \label{fig:cor-chain}
\end{figure}

\begin{proof}
    By \cref{cor:full-to-core}, there is a \graphalg{} Core expression $e$ that evaluates $f(A_1,\ldots,A_n)$, which by the definition of \graphalg{} Core is also valid muse-$\ML$ expression.
    Then by \cref{thm:muse-eq-dec}, there exists an expression $e^{*}$ in dec-$\ML$ that is equivalent to $e$.
    Furthermore, by \cref{thm:sifor-eq-dec} there exists an expression $e'$ in sifor-$\ML$ that simulates $e^{*}$.
    \Cref{fig:cor-chain} situates these languages in the wider hierarchy built on top of $\ML$.
\end{proof}

\section{Conclusion}
In this work, we have shown how \graphalg{} is derived from $\ML$ through a number of extensions.
We find that any \graphalg{} program can be simulated in sifor-$\ML$, an extension of for-$\ML$ that supports simultaneous extension.
The various other changes we have discussed in this work to distinguish \graphalg{} from sifor-$\ML$ do not influence the expressive power of the language.

\bibliographystyle{plain}
\bibliography{zotero}

\end{document}